\documentclass[conference]{IEEEtran}

\usepackage{subcaption}
\usepackage{algpseudocode} 
\usepackage{algorithm} 
\usepackage[utf8]{inputenc} 
\usepackage[T1]{fontenc}
\usepackage{url}
\usepackage{ifthen}
\usepackage{cite}
\usepackage[cmex10]{amsmath} 

\usepackage{cite}
\ifCLASSINFOpdf
 
\else
   \usepackage[dvips]{graphicx}
 \fi
\usepackage[cmex10]{amsmath}
\usepackage{amssymb}
\usepackage{stfloats}
\usepackage{tikz-cd}
\usepackage{amsthm} 
\usepackage{wrapfig}
\usepackage{amssymb}
\usepackage{latexsym}
\usepackage{bm, bbm} 
\usepackage{epsfig}
\usepackage{graphicx}
\usepackage{epsfig}
\usepackage{multicol}
\usepackage{psfrag}
\usepackage{float}
\usepackage{cite}
\usepackage{color}

\newtheorem{definition}{Definition}
\newtheorem{theorem}{Theorem}
\newtheorem{statement}{Statement}

\newtheorem{lemma}{Lemma}

\title{Topology-Aware Exploration of Energy-Based Models Equilibrium: Toric QC-LDPC Codes and Hyperbolic MET QC-LDPC Codes}

\author{%
  \IEEEauthorblockN{Vasiliy~Usatyuk}
  \IEEEauthorblockA{R\&D Department, T8\\ 
                    Email: L@lcrypto.com}
  \and
  \IEEEauthorblockN{Denis Sapozhnikov}
  \IEEEauthorblockA{R\&D Department, T8\\ 
                     Email: D@lcrypto.com}
   \and
  \IEEEauthorblockN{Sergey Egorov}
  \IEEEauthorblockA{Computer Science Department, South-West State University\\ 
                     Email: sie58@mail.ru}

                     }

\begin{document}
\maketitle
\begin{abstract}

This paper presents a 
method for achieving equilibrium in the ISING Hamiltonian when confronted with unevenly distributed charges on an irregular grid. Employing (Multi-Edge) QC-LDPC codes and the Boltzmann machine, our approach involves dimensionally expanding the system, substituting charges with circulants, and representing distances through circulant shifts. This results in a systematic mapping of the charge system onto a space, transforming the irregular grid into a uniform configuration, applicable to Torical and Circular Hyperboloid Topologies. The paper covers fundamental definitions and notations related to QC-LDPC Codes, Multi-Edge QC-LDPC codes, and the Boltzmann machine. It explores the marginalization problem in code on the graph probabilistic models for evaluating the partition function, encompassing exact and approximate estimation techniques. Rigorous proof is provided for the attainability of equilibrium states for the Boltzmann machine under Torical and Circular Hyperboloid, paving the way for the application of our methodology. Practical applications of our approach are investigated in Finite Geometry QC-LDPC Codes, specifically in Material Science. The paper further explores its effectiveness in the realm of Natural Language Processing Transformer Deep Neural Networks, examining Generalized Repeat Accumulate Codes, Spatially-Coupled and Cage-Graph QC-LDPC Codes. The versatile and impactful nature of our topology-aware hardware-efficient quasi-cycle codes equilibrium method is showcased across diverse scientific domains without the use of specific section delineations.
   \end{abstract}\vspace{-1mm}

%

\newcommand{\F}{\mathbb{F}_2}
\newcommand{\Z}{\mathbb{Z}}
\newcommand{\R}{\mathbb{R}}
\newcommand{\supp}{\mathrm{supp}\,}
\renewcommand{\P}{\mathbb{P}}
\renewcommand{\mod}{\, \mathrm{mod} \,}
\newcommand{\Wcal}{\mathcal{W}}

\section{Introduction}

The integration of graph models has had a transformative impact across diverse domains, playing a crucial role in fields such as forward error correction, source coding, compressed sensing, classical and quantum machine learning/computation, and control theory,  \cite{Ks01,Forney01,Forney13,Forney18}.

Despite notable successes in applying graph models to shallow networks and demonstrating achievements in lazy regime machine learning \cite{lazy}, challenges persist. These encompass issues like optimal neural network architectures, adversarial robustness in deep neural networks (DNNs), optimal preconditions for DNN training, efficient estimation of the Hessian matrix for quasi-Newton DNN training, and achieving sparse structural low-complexity learning and inference in DNNs. Additionally, there is a notable gap in establishing proofs for both constructive and non-constructive forms of Generalized Shannon's Theorem for nonlinear channels.

These challenges, spanning information theory, control theory, topology, machine learning, and cryptography, are effectively reduced to the framework of the generalized Poincaré theorem, specifically the Novikov conjecture \cite{Novikov1,Novikov2,James2000}. Through a topology surgery approach, involving a sufficient subspace metric (grid) to represent the $n$-dimensional manifold using low-dimensional decomposition (embedding), the paper tackles these challenges, \cite{Thurston1982,Sc83,Gri02,Gri03,Gri032}.

The residual error from topological surgery becomes a guiding force, leading toward the formulation and understanding of the generalized Shannon theorem. This holistic approach addresses challenges spanning various scientific domains, emphasizing the intricate connections between topology, information theory, and machine learning.

The paper aims to provide partial answers to questions by narrowing down the exploration of complex system dynamics to constructing topology-aware energy-based codes on graphs. It delves into the application of topology in information theory, particularly in Channel Coding, across diverse scientific domains like Material Science, Deep Neural Networks, Control Theory, and beyond. This dual perspective contributes to the advancement of information theory and broader interdisciplinary scientific endeavors. The study reveals that choosing Topology, whether Torical or Hyperbolic, for representing non-linear fields inherently gives rise to quasi-cyclic (Low density parity-check) codes on graphs. The codewords of these codes serve as energy minima, while pseudo-codewords form local minima around these symmetry points. The size of the circulant (Automorphism) of the quasi-cyclic code plays a pivotal role in determining the number of equivalent energy minima, unveiling a novel perspective on energy landscapes.

From a mathematical perspective, quasi-cyclic codes can be conceptualized as forming an irregular grid in space. Their inherent strong symmetry designates them as linearizer manifolds, crucial in studying the Hamiltonian and the conservation of physical quantities across spatial and temporal dimensions. Linearizers have the capability to transform an irregular grid into a regular one, as explored in works such as \cite{Kolmogorov1954,Mos62,Arn1963}.


Introducing a pioneering Topology-Aware method for achieving equilibrium in the Ising System Hamiltonian with unevenly distributed charges, the approach leverages dimensional expansion through the substitution of charges with circulants and distances with circulant shifts. This transformative process results in a spatial mapping that converts the initially irregular charge grid into a more uniform and structured configuration. The proposed methodology, with the invariance of Quasi-Cyclic rotations and the symmetry of codewords, coupled with the asymmetry of pseudocodewords, gives rise to structural (quasi-cyclic) sparsity, offering computational efficiency and storage optimization advantages. This makes the approach well-suited for applications dealing with large-scale Dynamic Systems.

The paper is structured as follows. In Section II, we present fundamental definitions and notations pertaining to QC-LDPC Codes, Multi-Edge QC-LDPC codes, and the Boltzmann machine (Ising model). Section III provides an overview of the marginalization problem in code on the graph probabilistic models, addressing the evaluation of the partition function (normalization). We explore exact estimation techniques using determinants and permanents, as well as approximate estimation employing the Bethe-permanent. Section IV is dedicated to proving that equilibrium states (ground states) of the Boltzmann machine, characterized by unequal charges on an irregular grid, are attainable under Torical and Circular Hyperboloid Topologies. This involves a systematic expansion of the system's charge dimension, with charges being replaced by multiple circulants and distances represented through circulant shifts. Section V delves into the practical application of this approach in Finite Geometry QC-LDPC Codes and its relevance to Material Science. Finally, Section VI investigates the application of this methodology to Generalized Repeat Accumulate Codes (\cite{DivMcel98,Li05}) and Cage-Graph QC-LDPC Codes (\cite{Ma07,Ke08,VeD08,Che15}), showcasing its effectiveness (LRA challenge, \cite{LRA21}) in the context of Natural Language Processing Transformer Deep Neural Networks Architectures: "MEGA" \cite{Ma22}, "ChordMixer" \cite{Kha22}, "CDIL" \cite{CDL}.

\section{Basic definition}

LDPC (Low-Density Parity-Check) codes, with parameters $[N, K]$, are linear error-correcting codes represented by sparse parity-check matrices. Here, $N$ is the length of the codeword, and $K$ is the number of information bits \cite{RR08,RyanShu09}.
The parity-check matrix $H$ specifies parity check equations and can be visualized as a Tanner graph. In the case of the parity-check matrix given by:

\begin{equation} 
H = \begin{bmatrix}  
1 & 0 & 1 & 1 & 1  \\
1 & 1 & 0 & 0 & 0  \\
0 & 1 & 1 & 1 & 1  \\ 
\end{bmatrix},
\end{equation} 

the corresponding Tanner graph is shown in Fig.~\ref{proto}, left.
 \begin{figure}
\centering
\includegraphics[width=\columnwidth]{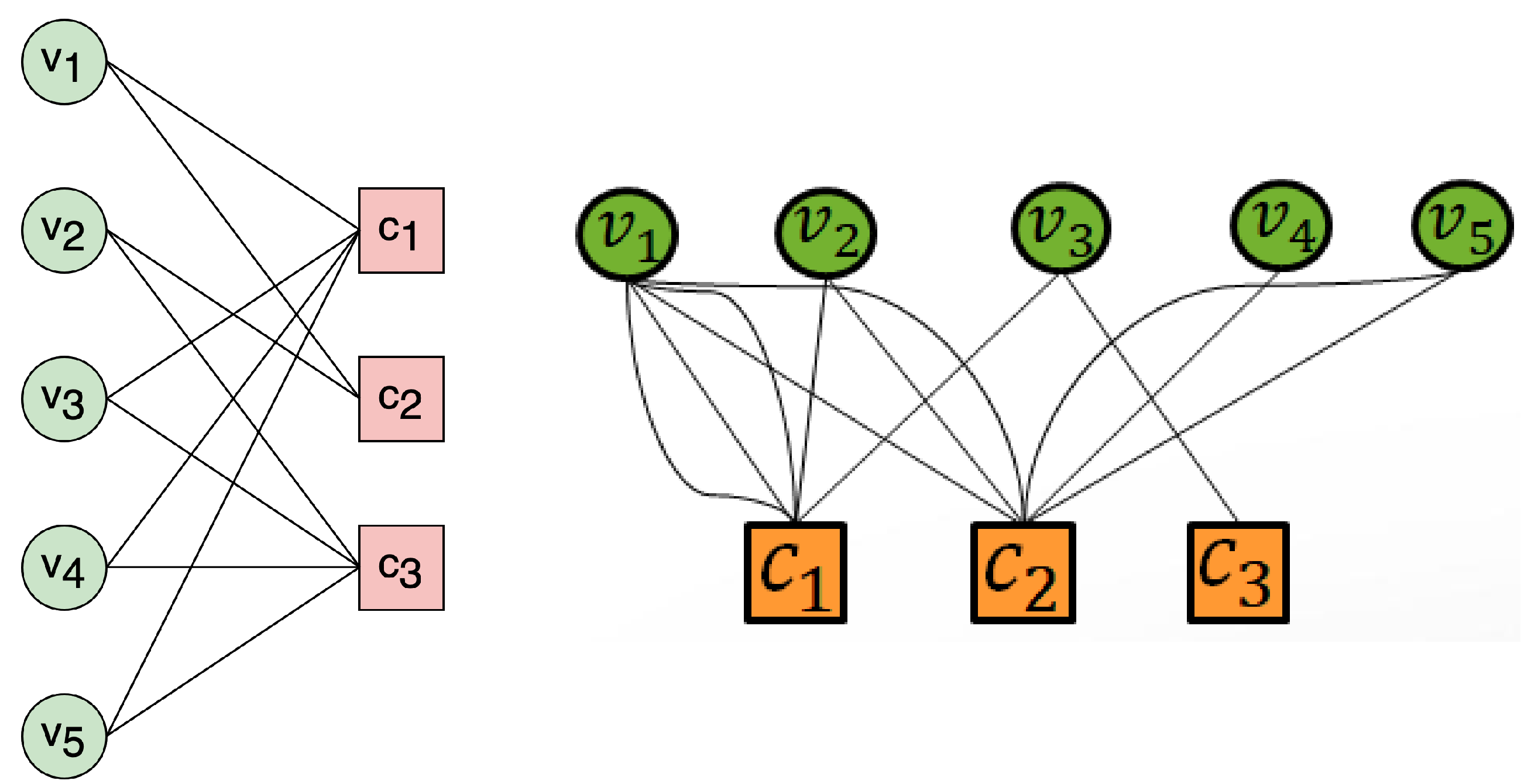} 
  \caption{(left) Tanner Graph of $H$ parity-check matrix (right) Multi-graph of $H_2$ parity-check matrix}
  \label{proto}
\end{figure}

Quasi-Cyclic Low-Density Parity-Check (QC-LDPC) codes, denoted as $[N, K]$, belong to the class of linear block error-correcting codes with a special quasi-cyclic parity-check matrix $H$. In an $(N, K)$ QC-LDPC code, $N$ represents the code length (number of codeword bits), and $K$ represents the number of message bits (length of the original message). The remaining bits in the codeword, i.e., $(N - K)$, are parity bits. The Tanner graph of a QC-LDPC code is described by the parity-check matrix $H$, consisting of square blocks that are either zero matrices or circulant permutation matrices.

The $L \times L$ circulant permutation matrix $P$ is defined as:

\begin{equation} 
P_{ij} = 
\begin{cases}
1, & \text{if } i+1 \equiv j \mod L \\
0, & \text{otherwise}.
\end{cases}
\end{equation} 

Here, $P_k$ represents the circulant permutation matrix (CPM) that shifts the identity matrix $I$ to the right by $k$ times for any $k$, $0 \leq k \leq L-1$. The zero matrix is denoted as $I^{-1}$, and the set $\{-1, 0, 1, \ldots, L-1\}$ is denoted by $A_{L}$.

Suppose the matrix $H$ of size $m L \times n L$ is defined as follows:

\begin{equation} 
H = \left[\begin{array}{cccc} {P_{a_{11}}} & {P_{a_{12}}} & {\cdots} & {P_{a_{1n}}} \\ {P_{a_{21}}} & {P_{a_{22}}} & {\cdots} & {P_{a_{2n}}} \\ {\vdots} & {\vdots} & {\ddots} & {\vdots} \\ {P_{a_{m1}}} & {P_{a_{m2}}} & {\cdots} & {P_{a_{mn}}} \end{array}\right] ,
\end{equation} 
where $a_{i,j} \in A_{L}$, and $L$ is the circulant size of $H$. A code $C$ with parity-check matrix $H$ is referred to as a QC-LDPC code. The exponent matrix $E(H) = (E_{ij}(H))$ of $H$ is given by:

\begin{equation} 
E(H) = \left[\begin{array}{cccc} {a_{11}} & {a_{12}} & {\cdots} & {a_{1n}} \\ {a_{21}} & {a_{22}} & {\cdots} & {a_{2n}} \\ {\vdots} & {\vdots} & {\ddots} & {\vdots} \\ {a_{m1}} & {a_{m2}} & {\cdots} & {a_{mn}} \end{array}\right],
\end{equation} 
where $E_{ij}(H) = a_{ij}$. The protograph mother matrix or base graph $M(H)$ is a binary matrix of size $m \times n$ obtained by replacing $-1$'s and other integers with $0$ and $1$, respectively, in $E(H)$.

The Tanner graph of matrix $H$ forms a cycle if the equation $\sum_{i=1}^{2l}(-1)^i a_i\equiv 0 \mod L$ is satisfied. A sub-graph of the Tanner graph, formed by cycles or cycle overlap in matrix $H$, includes $a$ variable nodes and $b$ odd-degree checks, known as a trapping set $TS(a,b)$, as defined in~\cite{Vasic09}. Fig. \ref{Trapping sets} illustrates examples such as TS(5,3) generated by the overlap of three 8-cycles and TS(4,4) formed by cycle 8 in the Tanner graph. The minimum codeword determining the code distance ($d_{min}$) of the LDPC code corresponds to $TS(a, 0)$, where $a=d_{min}$. Trapping sets form pseudocodewords ~\cite{Sm13, Sm15, Huang, Vo13, Vo22}.


\begin{figure}
\centering
\includegraphics[width=63mm, viewport=30.00mm 194.40mm 136.75mm 277.00mm]{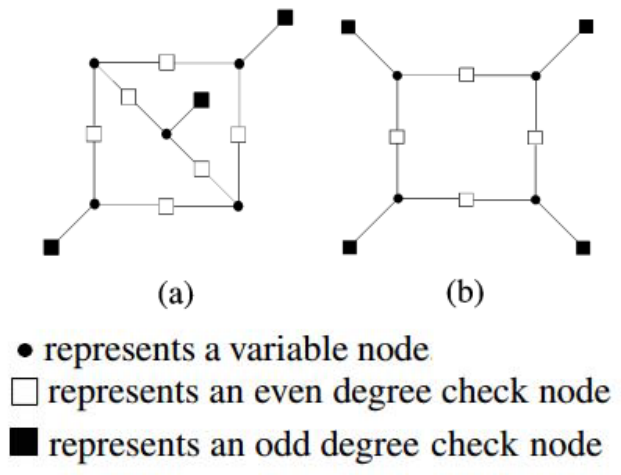} 
  \caption{Graphical representation of  Trapping sets: a) TS(5,3), b) TS(4,4), \cite{Vasic09}}
  \label{Trapping sets}
\end{figure}

\subsection{Multi-edge Type (Protograph) QC-LDPC Codes}

A QC-LDPC code and bipartite Tanner graph can be generalized to a multi-graph. The multi-graph corresponding to the parity-check matrix:

\begin{equation} 
H_2=\left(\begin{array}{ccccc} {I^{1} +I^{2} +I^{7} } & {I^{9} } & {I^{23} } & {0} & {0} \\ {I^{12} +I^{37} } & {I^{19} } & {0} & {I^{32} } & {I^{11} +I^{12} } \\ {0} & {0} & {I^{33} } & {0} & {0} \end{array}\right)
\end{equation} 
is shown in Figure \ref{proto} (right), where $I$ is the circulant permutation matrix used in multi-edge (protograph) QC-LDPC codes. In the multi-graph representation, the parity-check matrix is depicted as a collection of edges connecting vertices (circulants of weight more than 1). The multi-graph representation offers a more intuitive graphical perspective compared to hypergraph representation, as it directly shows the connections between variables and check nodes (constraints).

Multi-edge type low-density parity-check (LDPC) codes represent a class of error-correcting codes that utilize multiple types of edges to connect variable nodes and check nodes. Initially proposed by Richardson and Urbanke, these codes were later extended by various researchers. The degree distribution of multi-edge type LDPC codes is defined through a joint distribution of variable and check node degrees, with variable node degrees partitioned into multiple classes. This partitioning provides flexibility in code design, enhancing performance. Protographs are employed to visualize the structure of these codes, offering a graphical representation of the connectivity pattern between variable and check nodes.

Protographs were introduced by David MacKay and Radford M. Neal in the context of regular LDPC codes. Dariush Divsalar later proposed a protograph-based approach for designing multi-edge type LDPC codes. Divsalar's protograph simplifies the representation of code structure, facilitating analysis and design. The use of protographs has become a popular tool in the design and analysis of LDPC codes. The two approaches, Richardson-Urbanke's MET LDPC ~\cite{Ri02} and MacKay-Neal-Divsalar Protograph LDPC ~\cite{Di05}, represent different ways of describing the same type of LDPC code.

\subsection{Spatially-Coupled Convolutional QC-LDPC Codes}

The spatially-coupled MET QC-LDPC code is constructed by coupling together a series of $L$ disjoint block MET QC-LDPC codes into a single coupled chain through an edge-spreading operation \cite{FelZi99}, as depicted in Fig. \ref{Fig5}. A $(W, C, N, L)$ tail-biting spatial-coupled QC-LDPC code $C(H)$ of length $W$ and multiple $L$ can be defined by a parity-check base matrix:

\begin{equation} 
H=\left[\begin{array}{cccc} {A_{0,0} } & {A_{1,0} } & {...} & {A_{WL-1,0} } \\ {A_{0,1} } & {A_{1,1} } & {} & {A_{WL-1,1} } \\ {\vdots } & {\vdots } & {\ddots } & {\vdots } \\ {A_{0,CL-1} } & {A_{1,CL-1} } & {...} & {A_{WL-1,CL-1} } \end{array}\right]
\end{equation} 
where $0\le i\le CL-1$, $1\le j\le WL-1$, and CPM $A_{i,j}$ represents either an $N\times N$ zero matrix $Z$ or the $N\times N$ circulant permutation matrix $I\left(p_{i,j}\right)$ obtained by cyclically right-shifting the $N\times N$ identity matrix $I\left(0\right)$ by $p_{i,j}$ positions. Denote by CPM-column the i-th column of $A_{0,i},A_{1,i},\dots,A_{CL-1,i}$.

For a specific Tail-biting Spatial-couple MET QC-LDPC code, we define the corresponding CPM-shifts matrix as the matrix of circulant shifts that defines the QC-LDPC code:
\begin{equation} 
B=\left[\begin{array}{cccc} {b_{0,0} } & {b_{1,0} } & {...} & {b_{W-1,0} } \\ {b_{0,1} } & {b_{1,1} } & {} & {b_{W-1,1} } \\ {\vdots } & {\vdots } & {\ddots } & {\vdots } \\ {b_{0,C-1} } & {b_{1,C-1} } & {...} & {b_{W-1,C-1} } \end{array}\right].
\end{equation} 

Define vector D of shifts, Fig. \ref{Fig4}: $D^{T}=\left[ {d_{0} },  {d_{1} },  {\cdots },  {d_{C-1} }\right] ,$ with conditions $d_0=0$, $d_i<\ d_j\ {\rm iff\ }i<j,\ \forall i\ d_i<W$. If $\frac{W}{C}\ $ is an integer, vector $D$ can be defined as $d_i=\frac{i\ W}{C}{\rm \ }\ $.

For $0\le i\le CL-1$ and $0\le j\le WL-1$, CPMs $A_{i,j}$ in parity-check matrix $H$ can be calculated:
if $i'\ {\rm mod\ }\ WL\ \le i\ <\left(i'+W\right)\ {\rm mod\ }\ WL$, then $A_{i,j}=I(b_{i\ {\rm mod\ W}, j\ {\rm mod\ C\ }})$, else $A_{i,j}=Z$. Summary $C(H)$ is represented in Fig. \ref{Fig7}. In papers \cite{TanSSFCost04,JiaPsP06}, it was shown that QC-LDPC codes are equivalent to terminated Spatially-Coupled and Serial Turbo convolutional codes. Trapping sets pseudocodewords and related Bethe-permanent reach their upper bound \cite{SmaPuVon09,MiPuLeCo11,Sm13}, energy-landscape local minimum, as we shall show later.

\begin{figure}
\centering
\includegraphics[width=3.5in]{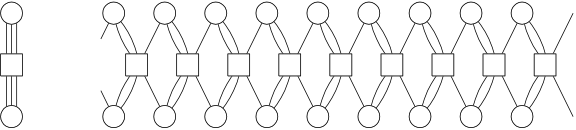}
  \caption{Example of protograph and spatial coupled set of protograph}
  \label{Fig5}
\end{figure}

\begin{figure}
\centering
\includegraphics[width=3.5in]{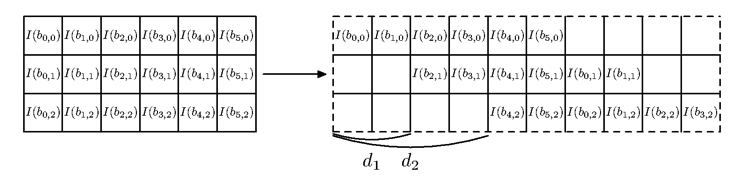}
  \caption{One $W \times C$ CPM-block of circulants $A(b_{(i,j)})$ shifted by values of vector D}
  \label{Fig4}
\end{figure}

\begin{figure}
\centering
\includegraphics[width=3.5in]{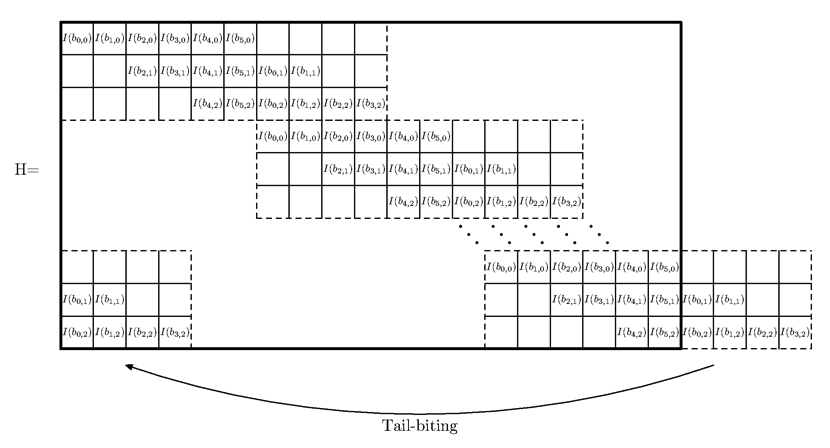}
  \caption{Tail-biting Spatial-Couple QC-LDPC code parity-check matrix $C(H)$}
  \label{Fig7}
\end{figure}

\subsection{Boltzmann machine}

A Boltzmann machine is a network comprising interconnected units (refer to Figure~\ref{fig:BM:BM}). Let \(N\) denote the number of units, each assuming a binary value (0 or 1). The random variable \(X_i\) represents the value of the \(i\)-th unit (\(i \in [1, N]\)). We use the column vector \(\boldsymbol{X}\) to denote the random values of all \(N\) units. The Boltzmann machine involves two types of parameters: bias (\(b_i\) for the \(i\)-th unit) and weight (\(w_{i,j}\) between units \(i\) and \(j\) for \((i, j) \in [1, N-1] \times [i+1, N]\)). The bias for all units is represented by the column vector \(\mathbf{b}\), and the weight for all unit pairs is represented by the matrix \(\mathbf{W}\), where the \((i, j)\)-th element of \(\mathbf{W}\) is \(w_{i,j}\), with \(w_{i,j} = 0\) for \(i \ge j\) and disconnected unit pairs. The parameters are collectively denoted as:

\begin{equation} 
\theta \equiv (b_1, \ldots, b_N, w_{1,2}, \ldots, w_{N-1,N}),
\end{equation} 
which is also expressed as \(\theta = (\mathbf{b}, \mathbf{W})\).

\begin{figure}[tb]
  \centering
  \includegraphics[width=0.5\linewidth]{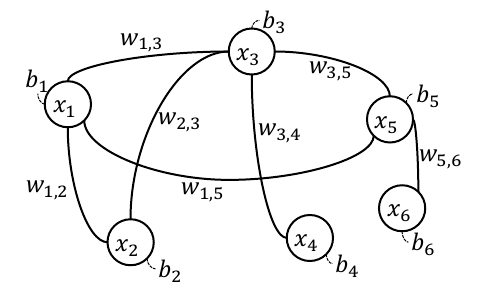}
  \caption{A Boltzmann machine, \cite{OsogamiIBM2}.}
  \label{fig:BM:BM}
\end{figure}

The Boltzmann machine's energy, denoted as \(E_\theta(\mathbf{x})\), is computed using the following expression:

\begin{multline}
E_\theta(\mathbf{x}) = - \sum_{i=1}^N b_i \, x_i - \sum_{i=1}^{N-1}\sum_{j=i+1}^N w_{i,j} \, x_i \, x_j = \\ - \mathbf{b}^\top \mathbf{x} - \mathbf{x}^\top \mathbf{W} \mathbf{x}.
\label{eq:BM:energy}
\end{multline}

Based on this energy formulation, the Boltzmann machine establishes the probability distribution over binary patterns through the expression:

\begin{equation} 
P_\theta(\mathbf{x}) = \frac{\exp\left( -E_\theta(\mathbf{x}) \right)}{\sum_{\mathbf{\tilde x}}\exp\left( -E_\theta(\mathbf{\tilde x}) \right)},
\label{eq:BM:prob}
\end{equation} 

where the summation with respect to \(\mathbf{\tilde x}\) spans all possible \(N\)-bit binary values. The denominator, referred to as the partition function, is denoted as:

\begin{equation} 
Z \equiv \sum_{\mathbf{\tilde x}} \exp\left( -E_\theta(\mathbf{x}) \right).
\end{equation}

A Boltzmann machine can be utilized to model the probability distribution \(P_{\rm target}(\cdot)\) for target patterns. By optimally configuring the values of \(\theta\), we can approximate \(P_{\rm target}(\cdot)\) with \(P_\theta(\cdot)\). In this context, hidden units, not directly associated with target patterns, can be introduced (refer to Figure~\ref{BMroles}, b). Units that directly correspond to target patterns are designated as visible, and these visible units can be further classified into input and output categories (see Figure~\ref{BMroles}, c). The Boltzmann machine is then capable of modeling the conditional distribution of output patterns given an input pattern.

\begin{figure*}[bt]
  \begin{minipage}[b]{0.3\linewidth}
    \includegraphics[width=\linewidth]{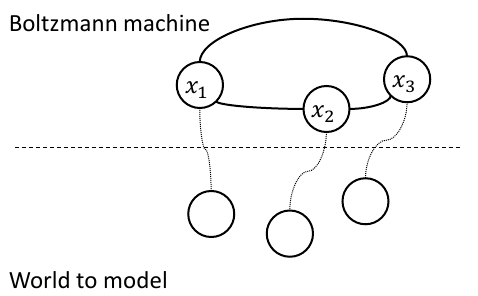}
    \subcaption{visible only}
    \label{BMvis}
  \end{minipage}
  \begin{minipage}[b]{0.3\linewidth}
    \includegraphics[width=\linewidth]{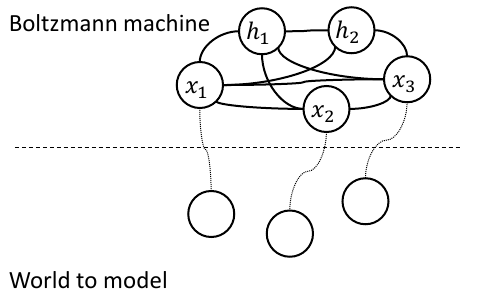}
    \subcaption{visible and hidden}
    \label{BMhidden}
  \end{minipage}
  \begin{minipage}[b]{0.3\linewidth}
    \includegraphics[width=\linewidth]{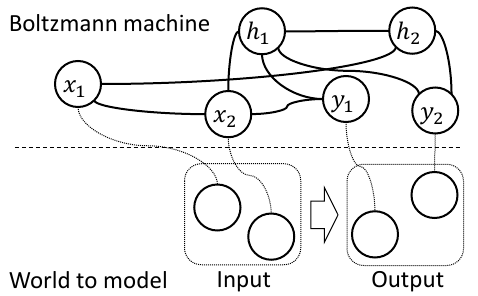}
    \subcaption{input and output}
    \label{BMinout}
  \end{minipage}
  \caption{Boltzmann machines with hidden units, input, and output, \cite{OsogamiIBM2}.}
  \label{BMroles}
\end{figure*}

\section{Marginalization Problem}

Probabilistic models often face challenges in computing the partition function, especially in normalizing to achieve a total probability of 1. To overcome this, probabilistic models can be transformed into energy-based models, eliminating the need for normalization and providing computational advantages. However, designing suitable objective functions for energy-based models remains a challenge, as discussed in \cite{YeFreWe05,KnFra20}.

In various Boltzmann machine training approaches, there exists an equation linking the probability (\(p\)) of the desired random variable distribution to the bond energy function (\(E\)) of input parameters, normalized by the known partition function (\(Z\)):

\begin{equation} 
p(x) = \frac{\exp(-E(x))}{Z}.
\end{equation} 

Different methods for determining the normalization factor are proposed. In Approach 1, the normalization factor relies on the determinant of the weight matrix \(W\):

\begin{equation} 
Z = \frac{1}{|\det(W)|}.
\end{equation} 

Another approach, as in \cite{Teh03}, also uses binding energy and determines the normalizing factor with a specific relation:

\begin{equation} 
Z(x;\theta) = \sum_{v}\sum_{h} \frac{1}{\exp(E(v,h,x;\theta))}.
\end{equation} 

For binding energy in a bipartite graph matrix, \cite{Teh03} proposes an analytical relation, called the marginal distribution:

\begin{multline}
E(v,x;\theta) = \sum_{j=1}^{|h|} \log[1+\exp(v^{T} W_{\bullet j}^{vh} +W_{j\bullet}^{vh} x+b_{j}^{h})] +\\
 v^{T} W_{}^{vx} x + v^{T} b^{v}.
\end{multline}

In the proposed solution, the energy determination is similar but numerical, derived from the Bethe approximation of Gibbs energy, \cite{Sm15,RoxanaPseudoBound,VontobelPseudoBound}:

\begin{equation} 
\exp(-E(\chi)) \sim \text{perm}_{\text{bethe}}(\chi).
\end{equation} 

Following the normalization condition, it is known that:

\begin{multline}
p(x) = \frac{\exp(-E(x))}{Z}, \\
\quad 0 \le p(x) \le 1 \Rightarrow Z = \sum \exp(-E(x)).
\end{multline}

These energy calculations are computationally feasible for sparse structural bipartite graphs, such as quasi-cyclic codes, and have been successfully employed in quasi-Newtonian methods for training deep neural networks via belief propagation, as demonstrated in \cite{Wei17, Teh03}. In particular, Bethe-energy learning could be accomplished using Bethe-permanent message-passing estimation, \cite{Huang, Vo13, Vo22}.

\section{Topology-Aware Equilibrium}

\begin{definition}\label{lrdist}
A particle is considered to be in equilibrium (ground state) if no force acts on it, or the sum of the forces acting on the particle is zero:

\begin{equation}
\label{eq:equilibrium}
\sum_{i=1}^{n} F_{i}^{q_{j}} = 0.
\end{equation}
\end{definition}\label{lrdist}

As geodesic lines on a torus are circles, we establish preliminary statements regarding circles, representing them as geodesic lines.
\begin{definition}\label{lrdist}

Let $d_{l}$ denote the left distance to the particle when walking counterclockwise from the left, and $d_{r}$ denote the right distance when walking counterclockwise from the right.

\end{definition}

\begin{lemma}
Two particles with the same charge on a circle are in equilibrium if the left distance between the particles equals the right distance.
\end{lemma}

\begin{proof}

According to the equilibrium criteria in Equation \eqref{eq:equilibrium}, a particle is in equilibrium when the net force acting on it is zero. Therefore,

\begin{equation}
\label{eq:equilibrium_condition}
F^{q_{1}} = F^{q_{2}} = \frac{kq_{1}q_{2}}{r_{r}^{2}} - \frac{kq_{1}q_{2}}{r_{l}^{2}} = 0 \Leftrightarrow r_{r} = r_{l}.
\end{equation}

\end{proof}

\begin{lemma}
$n$ identical particles on a circle are in a state of equilibrium if all distances between neighboring particles are equal.
\end{lemma}

\begin{proof}

By analogy with Lemma 1, we express the conditions for particle equilibrium:

\begin{align}
\label{eq:lemma2_conditions}
    F^{q_{1}} & = \frac{kq_{1} q_{2}}{d_{r1}^{2}} - \frac{kq_{1} q_{n}}{d_{l1}^{2}} = 0 \Leftrightarrow d_{r1} = d_{l1}, \nonumber \\
    F^{q_{2}} & = \frac{kq_{2} q_{3}}{d_{r2}^{2}} - \frac{kq_{2} q_{1}}{d_{l2}^{2}} = 0 \Leftrightarrow d_{r2} = d_{l2}, \, d_{r1} \equiv r_{l2}, \nonumber \\
    & \dots \nonumber \\
    F^{q_{n-1}} & = \frac{kq_{n-1} q_{n}}{d_{rn-1}^{2}} - \frac{kq_{n-1} q_{n-2}}{d_{l\, n-1}^{2}} = 0 \Leftrightarrow d_{rn-1} = d_{l\, n-1}, \nonumber \\
    & d_{rn-2} \equiv d_{ln-1}, \nonumber \\
    F^{q_{n}} & = \frac{kq_{n} q_{1}}{d_{rn}^{2}} - \frac{kq_{n} q_{n-1}}{d_{l\, n}^{2}} = 0 \Leftrightarrow d_{rn} = d_{l\, n}, \, d_{rn-1} \equiv d_{ln}.
\end{align}

These conditions are satisfied only when all distances are equal.
\end{proof}

The arrangement of particles on the circle is clearly a uniform grid. With these auxiliary statements, we are now prepared to prove the main theorem.

\begin{figure}[h]
    \centering
    \includegraphics[width=120pt]{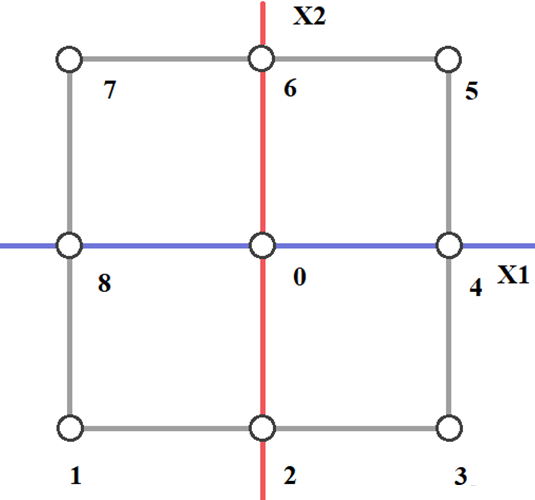}
    \caption{Quadrangular multicell of four cells on the surface of a torus.}
    \label{fig:quadrangular_multicell}
\end{figure}

\begin{theorem}\label{th1}
A collection of identical particles situated on the surface of a torus achieves equilibrium if and only if it configures into a uniform grid.
\end{theorem}

\begin{proof}

\textit{Necessity.} Initially, we establish that a set of particles in equilibrium on the torus surface adopts a uniform grid configuration. As per Lemma 2, the orthogonal nature of directions $x_{1}$ and $x_{2}$ on the torus ensures that a sufficiently large group of particles forms a uniform grid in each direction. Equilibrium conditions \eqref{eq:lemma2_conditions} are met for both orthogonal directions. Thus, the assertion stands proven.

\textit{Sufficiency.} Next, we demonstrate that particles positioned at the nodes of a uniform grid are, indeed, in a state of equilibrium. We examine a quadrangular multicell (Fig. \ref{fig:quadrangular_multicell}) on the torus surface, articulate the sum of forces acting on the nodal element, and confirm that it satisfies the equilibrium conditions.

\begin{figure*}[b]
\begin{equation} \label{eq:forces_projections} 
\begin{aligned}
    & X1: \left| F_{X1}^{q0} = kq_{0} \left(\frac{q_{1}}{r_{x10}^{2}} + \frac{q_{8}}{r_{x80}^{2}} + \frac{q_{7}}{r_{x70}^{2}} - \frac{q_{3}}{r_{x30}^{2}} - \frac{q_{4}}{r_{x40}^{2}} - \frac{q_{5}}{r_{x50}^{2}} \right) \right. = 0, \\
    & X2: \left| F_{X2}^{q0} = kq_{0} \left(\frac{q_{1}}{r_{y10}^{2}} + \frac{q_{2}}{r_{y20}^{2}} + \frac{q_{3}}{r_{y30}^{2}} - \frac{q_{5}}{r_{y50}^{2}} - \frac{q_{6}}{r_{y60}^{2}} - \frac{q_{7}}{r_{y70}^{2}} \right) \right. = 0.
\end{aligned}
\end{equation}.
\end{figure*}
Examine a quadrangular mesh. The force equations in axis projections are defined by equations \ref{eq:forces_projections}. It is clear that equilibrium conditions are fulfilled when either all projections of distances along one axis are identical or there is pairwise equality of projections with one having a positive sign and the other a negative sign. This results in solutions corresponding to either an orthogonal quadrilateral uniform mesh or a skewed quadrilateral uniform mesh, respectively.
\end{proof}
Similarly, the theorem can be demonstrated for a triangular uniform mesh. A more robust assertion states that these relationships are valid for a regular grid, as evident from Theorem 1 and non-constructive proof from \cite{Vuffray}.

Nevertheless, under the considered hypothesis, an alternative method for achieving equilibrium with unequal charges on an irregular grid is suggested. This approach entails a linear expansion of the charge system's dimension, substituting charges with multiple circulants and distances with circulant shifts. This results in a transformation of the charge system into a space where the charges constitute a uniform grid.

\begin{definition}\label{lrdist}
According to the Ising model, each cell of the lattice corresponds to a spin value equal to +1 or -1. Each arrangement of particles and spins is associated with energy:
\end{definition}

\begin{equation} 
E\left(\sigma \right) = \sum_{i,j} C_{ij} J_{ij} \sigma_{i} \sigma_{j},
\end{equation} 
where \(C_{ij}\) is the connectivity matrix of the spin element, and \(J_{ij}\) is the interaction energy between spins \(\sigma_{i}, \sigma_{j}\).

\begin{definition}\label{thoneline}
Energy potential is defined by considering the interaction of particles, all having charges of the same sign (for convenience of presentation), in the one-dimensional case (Fig. \ref{fig:three_particles_line}), using the Ising model.
\end{definition}

\begin{figure}[h]
    \centering
    \includegraphics[width=\columnwidth]{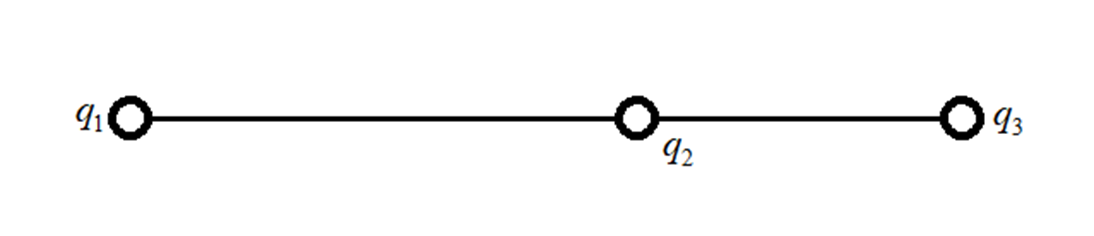}
    \caption{Three particles on one line.}
    \label{fig:three_particles_line}
\end{figure}

Charges 1 and 3, having the same sign, create an electrostatic field. The field energy created by charge \(i\) at point \(j\) is directly proportional to the magnitude of the charge and inversely proportional to the square of the distance between the points:

\begin{equation} 
J_{ij} = k\frac{q_{i}}{r_{ij}^{2}}
\end{equation} 

Quantizing charge and space, we define the values \(q_{1}\), \(q_{3}\), \(r_{12}^{2} = R_{1}\), and \(r_{32}^{2} = R_{3}\), where the electric field created by charge 1 at point 2 balances the field created at point 2 by charge 3. To determine possible states of the system where the field magnitude created by neighboring charges at the desired point reaches a minimum, we introduce the concept of a circulant \(I_{n}\)—a unit matrix of size \(e \times e\), subjected to a cyclic shift of magnitude \(n\).

\begin{statement}\label{state1}
Since the distances between particles are rational numbers, represented as proportional integers, we introduce the size of the circulant \(e\), considered as the maximum possible analogue of the distance between charges 1 and 3:
\end{statement}

\begin{equation} 
e = \sum (R_{1}, R_{3})
\end{equation} 

\begin{statement}\label{state1}
The magnitudes of the charges are multiples of the elementary charge, represented as an integer number of elementary charges. The integer ratios of possible charges to the square of the distance, not exceeding the size of the circulant, will be equal to the cyclic shift of the circulant:

\begin{equation} 
\frac{q_{i}}{R_{i}} = n, \, 0 \leq n < e.
\end{equation} 

Thus, the minimum field energy at the desired point forms a set of pairs of circulants of the form:

\[\begin{array}{l} {\frac{q_{1}}{R_{1}} = n_{1}; \frac{q_{3}}{R_{3}} = n_{3}}, \\ {n_{1} + n_{3} = e - 1}. \end{array} \]

The matrix describing the interaction of electric fields 1 and 3 at point 2, where the fields at point 2 are balanced, looks like this:
\end{statement}

\begin{equation} 
H = \left[I_{n_{1}} | 0 | I_{n_{3}}\right].
\end{equation}

In the case of three particles, the zero circulant can be reduced.

\textbf{Numerical Example 1.} The idea illustrating Statement is demonstrated for \(e = \sum (R_{1}, R_{3}) = 5\) in Figure \ref{fig:figure3}:

\begin{figure}[h]
    \centering
    \includegraphics[width=\columnwidth]{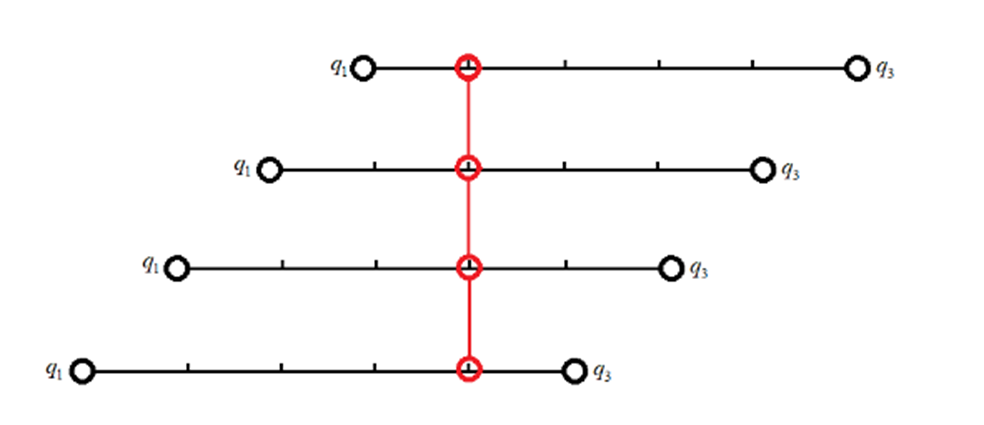}
    \caption{Coulomb's Law rules Quasi-Cyclic  \cite{Spa04}}
    \label{fig:figure3}
\end{figure}

In this example, we have a set of possible circulants that minimizes the field energy at point 2. The circulants are of the following form:

\begin{equation} 
H = \left(
\begin{array}{l}
    {\left[I_{0} | I_{5} \right]} \\
    {\left[I_{1} | I_{4} \right]} \\
    {\left[I_{2} | I_{3} \right]} \\
    {\left[I_{3} | I_{2} \right]} \\
    {\left[I_{4} | I_{1} \right]} \\
    {\left[I_{5} | I_{0} \right]} 
\end{array} \right).
\end{equation}

However, a circulant with zero shift (unit matrix) represents an equilibrium state of the system, corresponding to an infinite number of charge options. Let's choose from this set the states that can accurately describe the final system:

\begin{equation} 
H = \begin{array}{l}
    {\left(\begin{array}{l}
        {\left[I_{0} | I_{5} \right]} \\
        {\left[I_{1} | I_{4} \right]} \\
        {\left[I_{2} | I_{3} \right]} \\
        {\left[I_{3} | I_{2} \right]} \\
        {\left[I_{4} | I_{1} \right]} \\
        {\left[I_{5} | I_{0} \right]}
    \end{array}\right) \Rightarrow \left(\begin{array}{l}
        {\left[I_{1} | I_{4} \right]} \\
        {\left[I_{2} | I_{3} \right]} \\
        {\left[I_{3} | I_{2} \right]} \\
        {\left[I_{4} | I_{1} \right]}
    \end{array}\right)}
\end{array}.
\end{equation}

\begin{figure}[t]
	\begin{center}
            \includegraphics[width=\columnwidth]{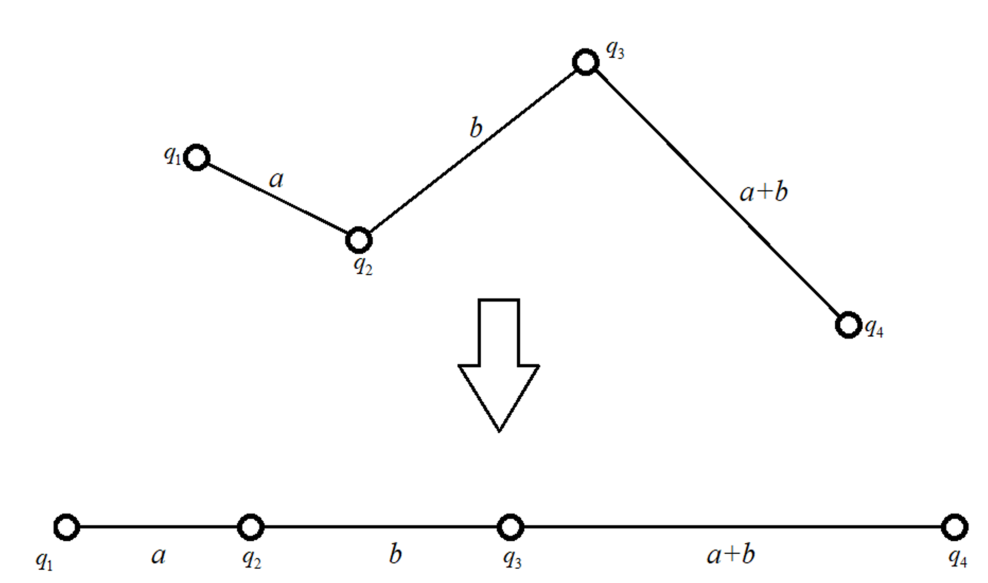}
	\end{center}
	\caption {Projection of a two-dimensional problem onto a one-dimensional case}
	\label{fig:images0001}
\end{figure}

\subsection{Two-dimensional Ising Equilibrium Problem}

Let's examine the scenario of four particles positioned on a straight line, representing the projection of the two-dimensional Ising problem. This is depicted in Figure \ref{fig:images0001}. Assume the squares of distances between particles are given by:
\begin{equation} 
\begin{array}{l} {r_{12}^2=a}, \\ {r_{23}^2=b}, \\ {r_{34}^2=a+b}. \end{array}
\end{equation} 
Here, $a$ and $b$ denote the squares of distances between the particles.

Due to the interaction screening (localization) of particles, particle 1 does not influence particle 3, and particle 4 does not affect particle 2. 

Now, let's explore solutions for two scenarios: the equilibrium/ground state of particles 1-3 and the equilibrium/ground state of particles 2-4.

For the 1-3 system, the matrix describing the interaction between charges can be expressed as:
\begin{equation} 
H =\left( 
\begin{array}{cccc} {I_{a}^{a+b}} & {0} & {I_{b}^{a+b}} & {0} \\ {0} & {I_{b}^{a+2b}} & {0} & {I_{a+b}^{a+2b}} \end{array} \right).
\end{equation}

This matrix comprises four circulants. The initial circulant signifies the interaction between particles 1 and 2, the second circulant characterizes the interaction between particles 2 and 3, the third circulant delineates the interaction between particles 3 and 4, and the fourth circulant illustrates the interaction between particles 1 and 4. In these expressions, $I_p^o$ represents a circulant of size $o$ with a cyclic shift of $p$. To standardize the circulants to the same size, we can utilize the following relations:
\begin{equation} 
\begin{array}{l} {c=\text{LCM}\left(a+b,a+2b\right)}, \\ {\overline{a}=\frac{a\cdot c}{a+b}; \overline{b_1}=\frac{b\cdot c}{a+b}}, \\ {\overline{b_2}=\frac{b\cdot c}{a+2b}; \overline{a+b}=\frac{\left(a+b\right)\cdot c}{a+2b}}. \end{array}
\end{equation} 
In these equations, $\text{LCM}(x,y)$ represents the least common multiple of $x$ and $y$. The bars over $a$, $b_1$, $b_2$, and $a+b$ denote the adjusted values.

By employing these correlations, we can derive a system of interaction among four particles in quasi-cyclic notation:
\begin{equation} H =\left( 
\begin{array}{cccc} {I_{\overline{a}}^{c}} & {I_{\overline{b_1}}^{c}} & {I_{\overline{b_2}}^{c}} & {I_{\overline{a+b}}^{c}} \end{array} \right) .
\end{equation} 
This matrix comprises four circulants, all standardized to the same size $c$, aligning with the minimum field energy at points 1-3 and 2-4.

{\flushleft \textbf{Numerical Example 2.}}
Consider an example with the following distances:
\begin{equation} 
\begin{array}{l} {r_{12}^{2} =2}, \\ {r_{23}^{2} =3}, \\ {r_{34}^{2} =5}. \end{array}
\end{equation} 

Describing the impact of particles 1 and 3 on particle 2, as well as particles 2 and 4 on particle 3, we obtain the following matrix:

\begin{equation} 
H =\left(  \begin{array}{cccc} {I_{2}^{5}} & {0} & {I_{3}^{5}} & {0} \\ {0} & {I_{3}^{8}} & {0} & {I_{5}^{8}} \end{array} \right).  
\end{equation} 

This matrix comprises four circulants with sizes 5 and 8. The first circulant represents the interaction between particles 1 and 2, the second circulant represents the interaction between particles 2 and 3, the third circulant represents the interaction between particles 3 and 4, and the fourth circulant represents the interaction between particles 1 and 4.

Note that the circulants have different sizes due to the need for adjustment to ensure they all have the same size. In this case, the adjusted sizes are 5 and 8. This matrix contains information about possible charge configurations resulting in the minimum field energy at points 2-3.

To standardize the description to circulants of the same size, we can use the following relationships:
\begin{equation} 
\begin{array}{l} {c=\text{LCM}(5,8)=40}, \\ {\overline{2}=\frac{2\cdot c}{2+3}=16; \overline{3}=\frac{3\cdot c}{2+3}=24}, \\ {\overline{5}=\frac{5\cdot c}{3+5}=15; \overline{2+3}=\frac{(2+3)\cdot c}{3+5}=25}. \end{array}
\end{equation} 
Using these relationships, we can obtain a parity-check matrix with a circulant size of 40 and cyclic shift values of 16, 15, 24, and 25:
\begin{equation}  H =\left( 
\begin{array}{cccc} {I_{16}^{40}} & {I_{15}^{40}} & {I_{24}^{40}} & {I_{25}^{40}} \end{array} \right).
\end{equation} 

In summary, we can apply a similar approach to the one-dimensional case to find possible charge configurations resulting in the equilibrium state of the system with the lowest field energy magnitude at a desired point. We can construct a parity-check matrix describing charge interactions in the two-dimensional case using distances between particles and insights gained from the one-dimensional case.


\begin{figure}[t]
	\begin{center}
            \includegraphics[width=\columnwidth]{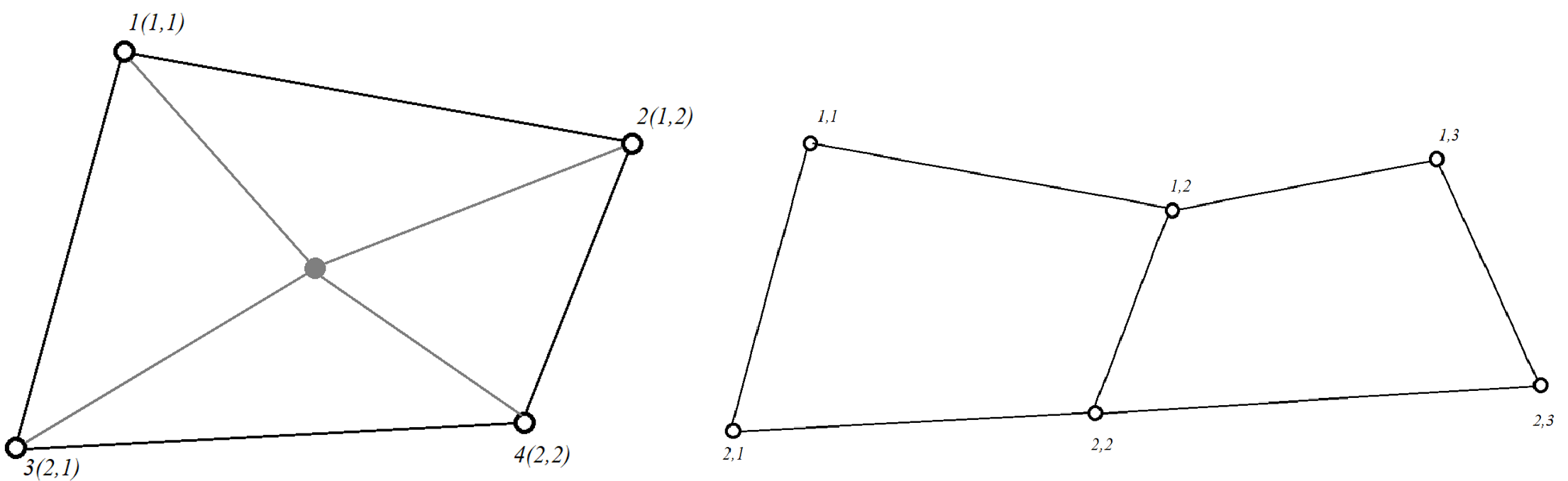}
	\end{center}
	\caption {(left) Interaction of 4 particles in 2d space. (right) Interaction of 6 particles in 2d space}
	\label{fig:images0011111}
\end{figure}

\subsection{Interaction of Four Particles in Two-Dimensional Space Without Projecting Onto One Axis}

Now, let's explore the interaction of four particles in two-dimensional space without projecting them onto one axis (illustrated in Fig.~\ref{fig:images0011111}, left).

For four particles in two-dimensional space, we can create a parity-check matrix with two circulant rows: one for interaction along the X-axis and one for interaction along the Y-axis. The order of each circulant is determined by the sum of the projections of distances between particles onto the respective axis. We define these values as:
\begin{equation} 
\begin{array}{l} {L_{x} =\sum _{i=1}^{4}L_{x_{i} x^{*} }  ;L_{x_{i} x^{*} } =L_{x_{i} } }, \\ {L_{y} =\sum _{i=1}^{4}L_{y_{i} y^{*} } ;L_{y_{i} y^{*} } =L_{y_{i} }  }. \end{array}
\end{equation} 

We then define an interaction matrix as:
\begin{equation} 
H^{j} =\left(\begin{array}{cccc} {I_{L_{x1}^{j} }^{L_{x}^{j} } } & {I_{L_{x2}^{j} }^{L_{x}^{j} } } & {I_{L_{x3}^{j} }^{L_{x}^{j} } } & {I_{L_{x4}^{j} }^{L_{x}^{j} } } \\ {I_{L_{y1}^{j} }^{L_{y}^{j} } } & {I_{L_{y2}^{j} }^{L_{y}^{j} } } & {I_{L_{y3}^{j} }^{L_{y}^{j} } } & {I_{L_{y4}^{j} }^{L_{y}^{j} } } \end{array}\right).
\end{equation} 

However, we may need to adjust the sizes of the circulants in the interaction matrix to ensure they are all the same. This can be achieved by finding the least common multiple of $L_{x}^{j}$ and $L_{y}^{j}$ and scaling the sizes of the circulants accordingly:
\begin{equation} 
\begin{array}{l} {L^{j} =LCM\left(L_{x}^{j} ,L_{y}^{j} \right)}, \\ {\overline{L_{xi}^{j} }=L_{xi}^{j} \cdot \frac{L^{j} }{L_{x}^{j} } ;\overline{L_{yi}^{j} }=L_{yi}^{j} \cdot \frac{L^{j} }{L_{y}^{j} } }. \end{array}
\end{equation} 

The final form of the interaction matrix will be:
\begin{equation} 
H^{j} =\left(\begin{array}{cccc} {I_{\overline{L_{x1}^{j} }}^{L_{}^{j} } } & {I_{\overline{L_{x2}^{j} }}^{L_{}^{j} } } & {I_{\overline{L_{x3}^{j} }}^{L_{}^{j} } } & {I_{\overline{L_{x4}^{j} }}^{L_{}^{j} } } \\ {I_{\overline{L_{y1}^{j} }}^{L_{}^{j} } } & {I_{\overline{L_{y2}^{j} }}^{L_{}^{j} } } & {I_{\overline{L_{y3}^{j} }}^{L_{}^{j} } } & {I_{\overline{L_{y4}^{j} }}^{L_{}^{j} } } \end{array}\right).
\end{equation} 

This matrix describes the interaction between charges in two-dimensional space and contains information about possible charge configurations resulting in the lowest magnitude of the field energy at a desired point.

\subsection{Interaction of Six Particles in Two-Dimensional Space Without Projecting Onto One Axis}

Let's delve into the interaction scenario involving six particles in two-dimensional space, where a pair of particles participates in both \(j\) and \(j+1\) cells (as depicted in Fig. \ref{fig:images0011111}, b). The parity-check matrices for these two adjacent cells can be defined as follows:

\begin{multline}
  H^{j} = \begin{pmatrix} I_{\overline{L_{x1}^{j}}}^{L_{}^{j}} & I_{\overline{L_{x2}^{j}}}^{L_{}^{j}} & I_{\overline{L_{x3}^{j}}}^{L_{}^{j}} & I_{\overline{L_{x4}^{j}}}^{L_{}^{j}} \\ I_{\overline{L_{y1}^{j}}}^{L_{}^{j}} & I_{\overline{L_{y2}^{j}}}^{L_{}^{j}} & I_{\overline{L_{y3}^{j}}}^{L_{}^{j}} & I_{\overline{L_{y4}^{j}}}^{L_{}^{j}} \end{pmatrix},  
  \\
  H^{j+1} = \begin{pmatrix} I_{\overline{L_{x2}^{j}}}^{L_{}^{j}} & I_{\overline{L_{x2}^{j}}}^{L_{}^{j+1}} & I_{\overline{L_{x4}^{j}}}^{L_{}^{j}} & I_{\overline{L_{x4}^{j}}}^{L_{}^{j+1}} \\ I_{\overline{L_{y2}^{j}}}^{L_{}^{j}} & I_{\overline{L_{y2}^{j+1}}}^{L_{}^{j+1}} & I_{\overline{L_{y4}^{j}}}^{L_{}^{j}} & I_{\overline{L_{y4}^{j+1}}}^{L_{}^{j+1}} \end{pmatrix}.
\end{multline}

The columns of each matrix represent the contributions of particles when listed with row priority. This approach can be extended to form quadrangular or triangular cells in four directions using sets of particles. The spatial configuration of particle locations is toroidal, enabling circular bypass of the configuration both "from bottom to top" and "from left to right" from a physical perspective on the boundary conditions.

\begin{statement}\label{state1} 
For a system of interacting particles in a two-dimensional space described by the Ising model, with charges having the same sign and creating an electric (magnetic) field, and distances between particles given by the sum of projections on the X and Y axes, the interaction matrix will have two circulant rows: interaction along the X axis and along the Y axis. The matrix describing the interaction of particles can be defined as a pair of circulant matrices that correspond to a set of pairs of circulants. To bring the orders of the circulants to a common value, the least common multiple of the distances is used to adjust the sizes.
\end{statement}

 
\begin{figure}[t]
	\begin{center}
            \includegraphics[width=100pt]{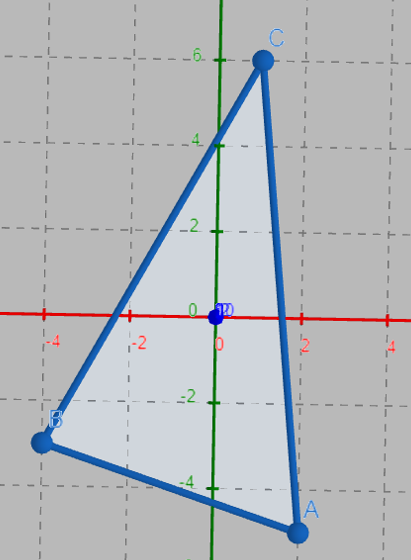}
	\end{center}
	\caption {Interaction of three particles on a plane. Example 2 from article \cite{TanSSFCost04}}
	\label{images0016}
\end{figure}

\begin{figure}[t]
	\begin{center}
            \includegraphics[width=\columnwidth]{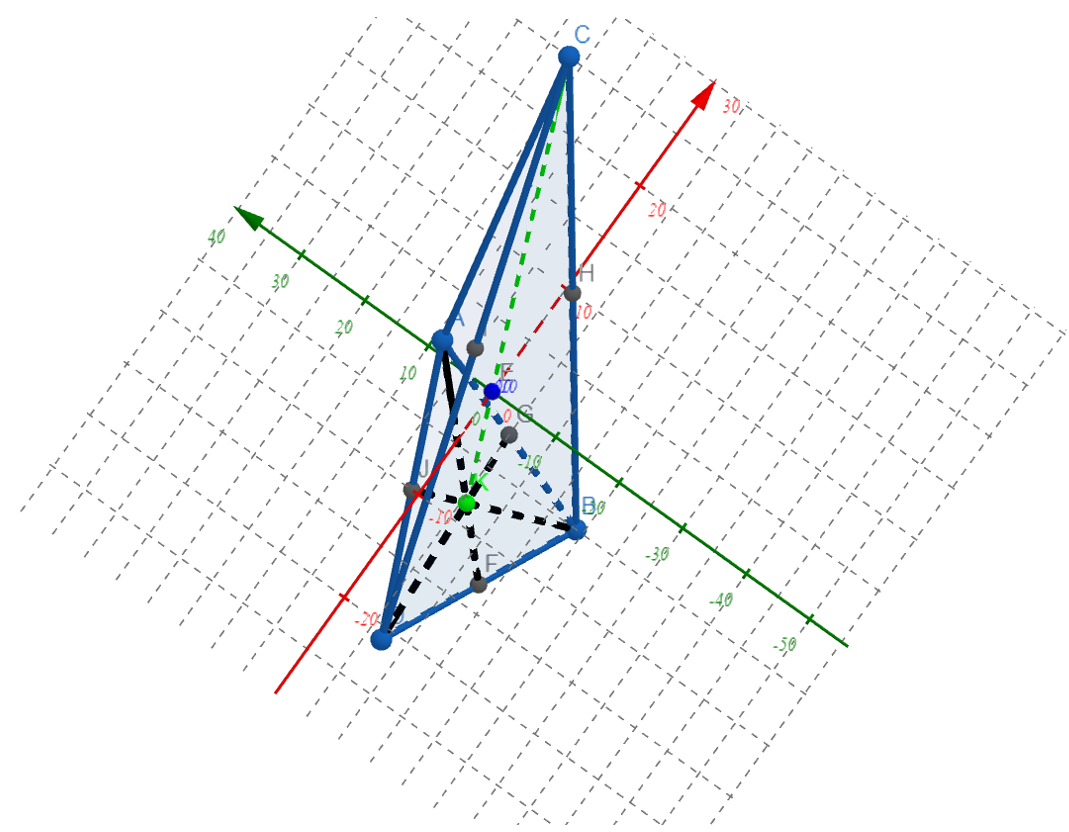}
	\end{center}
	\caption {Interaction of 4 particles in space. Example 3 from \cite{TanSSFCost04}}
	\label{fig:images0017}
\end{figure}

\begin{figure}[t]
	\begin{center}
            \includegraphics[width=\columnwidth]{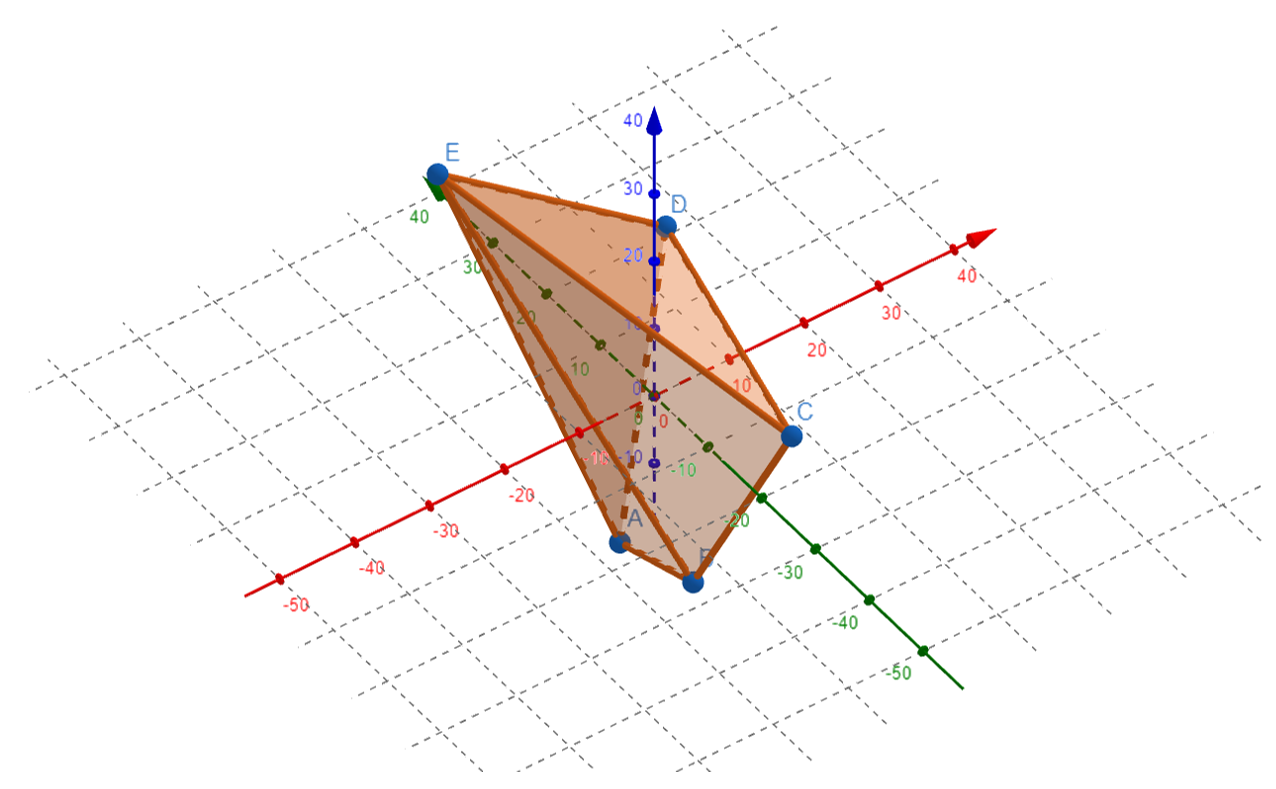}
	\end{center}
	\caption {Interaction of five particles in  space }
	\label{fig:images0019}
\end{figure}
\begin{figure}[t]
	\begin{center}
            \includegraphics[width=100pt]{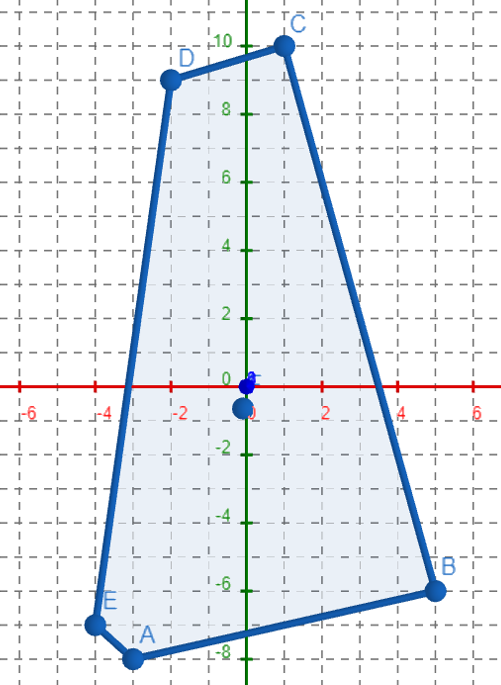}
	\end{center}
	\caption {Interaction of five particles in a plane}
	\label{fig:images0020}
\end{figure}

\textbf{Numerical Examples Using QC-LDPC Codes}

In the paper \cite{TanSSFCost04}, quasi-cyclic block and convolution LDPC codes based on circulant matrices are explored. Their geometric representations can be constructed using the proposed approach. The elementary codes from this paper, based on 2x3, 3x4, and 3x5 circulants, can be depicted as a set of particles in Euclidean space.

For instance, consider the code with the parity-check matrix:
\begin{equation} 
H = \begin{pmatrix}
I_1^7 & I_2^7 & I_4^7 \\
I_6^7 & I_5^7 & I_3^7
\end{pmatrix},
\end{equation} 
representing the interaction of three particles on a plane (as shown in Fig.~\ref{images0016}). Each row corresponds to a particle, and each column corresponds to a check node. The identity matrices $I_i^7$ represent connections between particles and check nodes, arranged in a circulant pattern. The position of the center of mass of a triangular plate with particles at the vertices aligns closely with the origin. In two-dimensional toric space, an analogue for such a configuration is an isosceles right triangle with vertices equidistant from the origin.

Another QC code with the parity-check matrix (\cite{TanSSFCost04}), Ex. 3): 
\begin{equation} 
H=\left(\begin{array}{cccc} {I_{1}^{26} } & {I_{5}^{26} } & {I_{25}^{26} } & {I_{21}^{26} } \\ {\begin{array}{c} {I_{9}^{26} } \\ {I_{3}^{26} } \end{array}} & {\begin{array}{c} {I_{19}^{26} } \\ {I_{15}^{26} } \end{array}} & {\begin{array}{c} {I_{17}^{26} } \\ {I_{23}^{26} } \end{array}} & {\begin{array}{c} {I_{11}^{26} } \\ {I_{11}^{26} } \end{array}} \end{array}\right),
\end{equation} 
can be represented as a system of four particles in space  (in~Fig.~\ref{fig:images0017}).

Four particles form a tetrahedron with the center of mass coinciding with the origin. In toric three-dimensional space, this configuration corresponds to four points of the cube, forming a trihedral angle, equidistant from the origin, with a parity-check matrix: 
\begin{equation} 
H=\left(\begin{array}{cccc} {I_{1}^{31} } & {I_{2}^{31} } & {I_{4}^{31} } & {\begin{array}{cc} {I_{8}^{31} } & {I_{16}^{31} } \end{array}} \\ {\begin{array}{c} {I_{5}^{31} } \\ {I_{25}^{31} } \end{array}} & {\begin{array}{c} {I_{10}^{31} } \\ {I_{19}^{31} } \end{array}} & {\begin{array}{c} {I_{20}^{31} } \\ {I_{7}^{31} } \end{array}} & {\begin{array}{c} {\begin{array}{cc} {I_{9}^{31} } & {I_{18}^{31} } \end{array}} \\ {\begin{array}{cc} {I_{14}^{31} } & {I_{28}^{31} } \end{array}} \end{array}} \end{array}\right),
\end{equation} 
describing the interaction of five particles in space. The particle configuration defining the code is shown in~Fig.~\ref{fig:images0019}.

Consider the following parity-check matrix
\begin{equation}  H=\left(\begin{array}{ccccc} {I_{10}^{11} } & {I_{9}^{11} } & {I_{8}^{11} } & {I_{7}^{11} } & {I_{6}^{11} } \\ {I_{1}^{11} } & {I_{2}^{11} } & {I_{3}^{11} } & {I_{4}^{11} } & {I_{5}^{11} } \end{array}\right),
\end{equation} 
describing the interaction of five particles in a plane. A polygon with charges at the vertices is shown in (in~Fig.~\ref{fig:images0020}).


\subsection{Circular Hyperboloid}

Now, let's explore a system of particles arranged in two belts on a cylindrical surface (Figure 4a). The connections between corresponding particles on neighboring belts align with the generatrices of the cylinder:

\begin{equation} 
\left\{
    \begin{array}{l}
        {\frac{x^{2}}{a^{2}} + \frac{y^{2}}{a^{2}} = 1} \\
        {-\infty < z < \infty}
    \end{array},
\right.
\end{equation} 
forming a regular grid.

For example, if a disturbance is introduced into the upper belt (Figure 4b), causing a displacement of particles at a certain angle $\alpha$, the particles will maintain their regular grid formation. The system of connections between corresponding particles remains linear; however, the connections no longer form a cylinder but rather a single-sheet hyperboloid, corresponding to the equation:

\begin{equation} 
\frac{x^{2}}{a^{2}} + \frac{y^{2}}{a^{2}} - \frac{z^{2}}{f(\alpha)} = 1,
\end{equation} 
where $f(\alpha)$ is a function dependent on the magnitude of the energy disturbance introduced into the system.

\begin{figure}[h]
    \centering
    \includegraphics[width=\columnwidth]{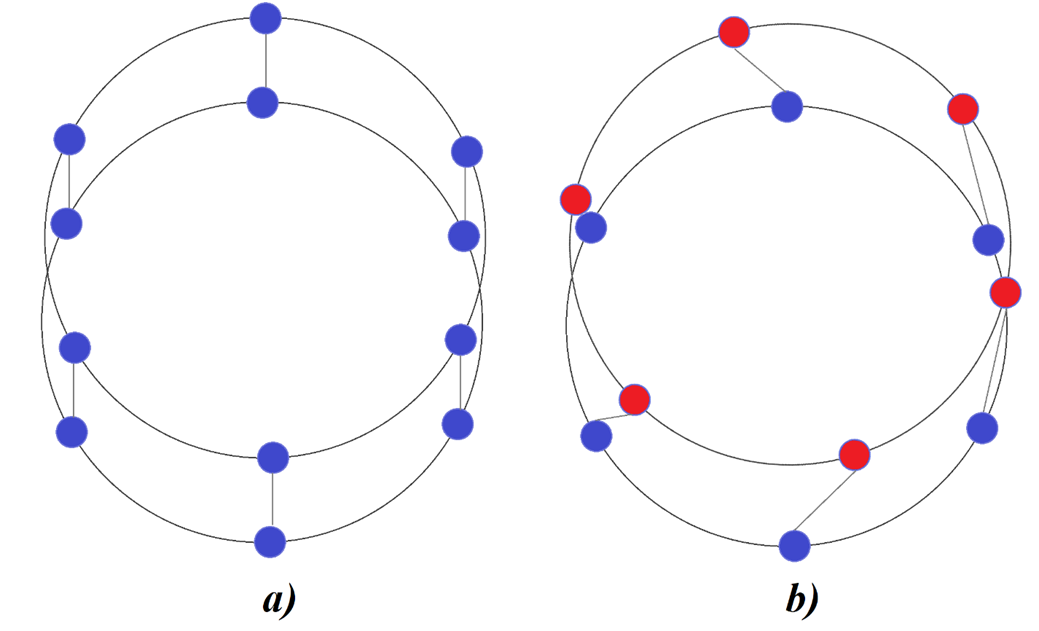}
    \caption{Figure 4. Charges on a) cylindrical and b) hyperbolic surfaces.}
    \label{fig:figure4}
\end{figure}

\begin{figure}[t]
	\begin{center}
            \includegraphics[width=\columnwidth]{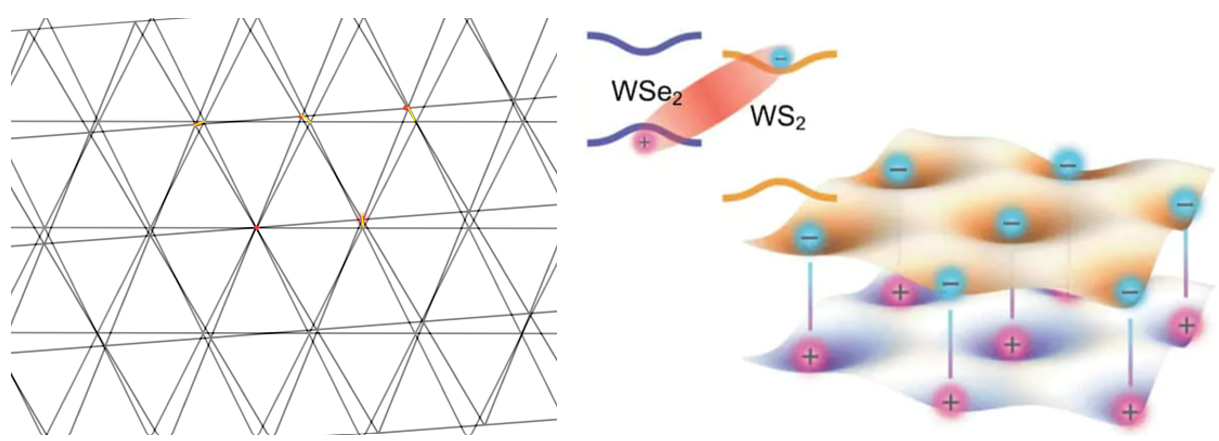}
	\end{center}
	\caption {Fermions are held in shifted lattices by a certain angle  (left),  regular triangular shape made of wolfram disulfide and wolfram diselenide from \cite{Xio23}  }
	\label{fig:FermBosLat}
\end{figure}

\begin{figure}[t]
	\begin{center}
            \includegraphics[width=3.2in]{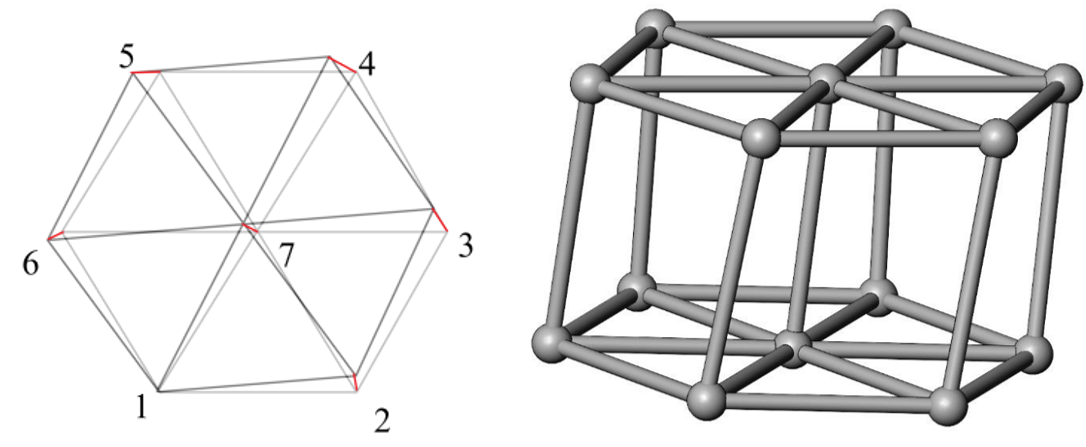}
	\end{center}
	\caption { The 7 dipoles hexagonal prism   }
	\label{fig:HexPrism}
\end{figure}

\begin{figure}[t]
	\begin{center}
            \includegraphics[width=2.6in]{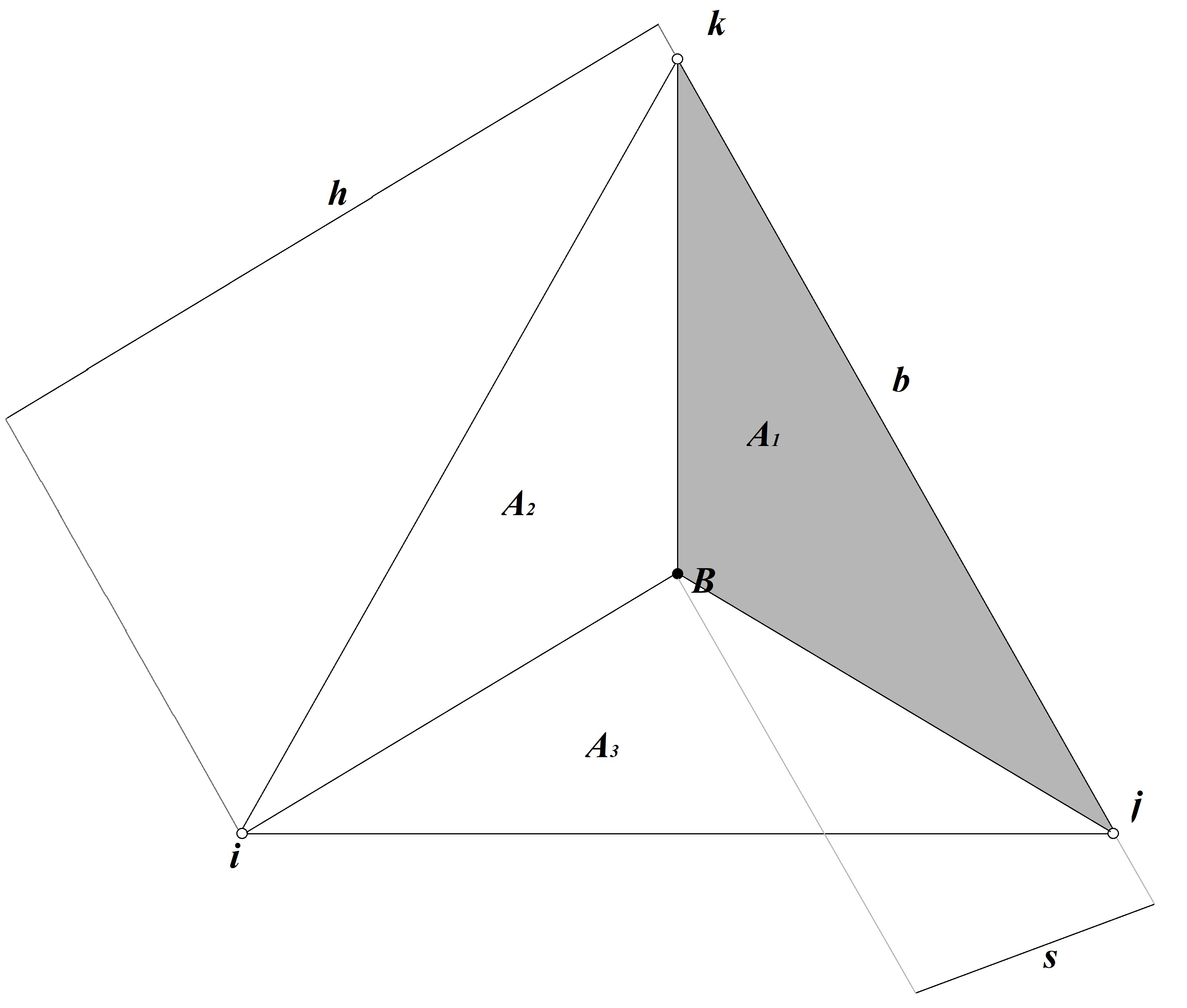}
	\end{center}
	\caption { Finite Element Analysis local coordinate system grid from \cite{Sege76}  }
	\label{fig:SegLoc}
\end{figure}

\begin{figure}[t]
	\begin{center}
            \includegraphics[width=1.8in]{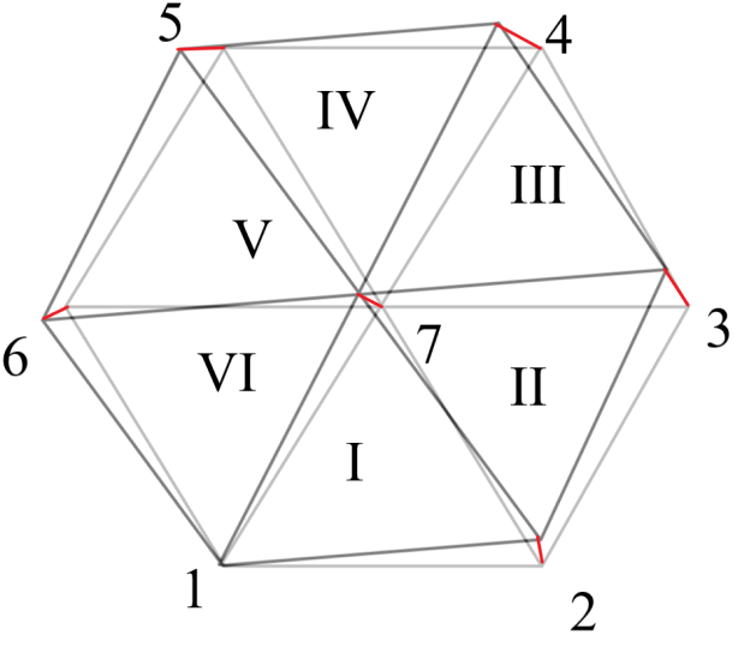}
	\end{center}
	\caption { Grid discretization of interaction  }
	\label{fig:FBcells}
\end{figure}

\subsection{Example of Circular Hyperboloid: MET QC-LDPC Codes}

The study in \cite{Xio23} investigates an experiment involving two lattices with regular triangular shapes made of wolfram disulfide and wolfram diselenide, each containing negatively charged electrons and positively charged electron holes in the nodes. Fermions are held in lattice nodes through strong interaction. When the lattices (Fig. \ref{fig:FermBosLat}) are shifted by a certain angle $\alpha$ relative to the stationary center of rotation, and photon energy is pumped, fermions enter into electromagnetic interaction, forming dipoles-excitons. Due to their dipole nature, excitons act as dielectric insulators, maintaining a strict lattice structure through strong interaction.

Consider a hexagonal prism containing 7 dipoles, one of which serves as the center of rotation for the grid (Fig. \ref{fig:HexPrism}). By introducing a Cartesian coordinate system centered at the grid's rotation center, we can derive a matrix describing the effect of grid displacement on particle interaction at a fixed center. Let $R$ be the step of the triangular grid, and $L_{ij}$ represent the distance between two nodes. Then,

\begin{equation} \label{1.21)} 
\begin{array}{l} {\Delta R_{jx} = L_{1jx} - L_{1jx} \cos \alpha }, \\ {\Delta R_{jy} = L_{1jy} - L_{1jy} \sin \alpha }. \end{array} 
\end{equation} 

Introduce quantization:

\begin{equation} \label{1.22)} 
{dR = \left\lfloor \frac{R}{q} \right\rfloor},  {qdR = \overline{R}}, {\overline{\Delta R} = \left\lfloor \frac{\Delta R}{dR} \right\rfloor}. 
\end{equation} 

Then, we can describe the interaction matrix $H$ as follows:

{\footnotesize
\begin{equation}\label{1.23)} 
\left(\begin{array}{ccccccc} {C_{0}^{q}} & {C_{\overline{\Delta R_{2x}}}^{q}} & {C_{\overline{\Delta R_{3x}}}^{q}} & {C_{\overline{\Delta R_{4x}}}^{q}} & {C_{\overline{\Delta R_{5x}}}^{q}} & {C_{\overline{\Delta R_{6x}}}^{q}} & {C_{\overline{\Delta R_{7x}}}^{q}} \\ {C_{0}^{q}} & {C_{\overline{\Delta R_{3y}}}^{q}} & {C_{\overline{\Delta R_{3y}}}^{q}} & {C_{\overline{\Delta R_{4y}}}^{q}} & {C_{\overline{\Delta R_{5y}}}^{q}} & {C_{\overline{\Delta R_{6y}}}^{q}} & {C_{\overline{\Delta R_{7y}}}^{q}} \end{array}\right) .
\end{equation}
}

At a sufficient distance from the center of rotation, the bevel of the dipoles may exceed the value at which the connection between the i-th electrons and the holes of the corresponding layers is preserved. In this case, the influence passes to neighboring nodes, changing the direction of the circulant shift. Additionally, three main directions can be distinguished on such a grid, and a local non-orthogonal system of three coordinates L1-3 can be introduced. Coordinates vary from 0 to 1 (exceeding 1 when the grid is beveled), \cite{Sege76}.

The local coordinate system is connected by perpendiculars, projecting from a point to a side. Each coordinate expresses the ratio of the area of the triangle constructed on the corresponding side and point to the area of the original triangle (Fig. \ref{fig:SegLoc}). The coordinate of point B in the L-system will have the value:

\begin{equation} \label{1.24)} 
\begin{array}{l} {B_{L} = (L_{1}, L_{2}, L_{3}) = \left(\frac{S_{\Delta Bjk}}{S_{\Delta ijk}}, \frac{S_{\Delta Bki}}{S_{\Delta ijk}}, \frac{S_{\Delta Bij}}{S_{\Delta ijk}}\right)} \\ {=\left(\frac{h_{B1}}{h_{1}}, \frac{h_{B2}}{h_{2}}, \frac{h_{B3}}{h_{3}}\right)} .\end{array}
\end{equation}

Each displaced lattice node participates in the formation of six grid cells and can be represented, according to the proposed hypothesis, by six sets of columns of three circulants describing the shift of the node relative to the local coordinates of each of the 6 cells. By the properties of additivity, the columns describing the interaction in each cell will be the sum of three columns describing the shifts of each of the cell nodes relative to the local coordinate system.

\section{Implementing the Proposed Approach in Finite Geometry QC-LDPC Codes}

Examine the QC-LDPC code introduced in \cite{Fos01}, featuring mixed circulant sizes with various automorphism structures:

\begin{multline}
\label{1.1)} 
H=\left(\begin{array}{c} {\begin{array}{ccccc} {I_{17} } & {I_{17} } & {I_{17} } & {I_{17} } & {I_{17} } \end{array}} \\ {H_{1} } \\ {H_{2} } \\ {H_{3} } \\ {H_{4} } \end{array}\right)=  \\
\left(\begin{array}{c} {\begin{array}{ccccc} {I_{17} } & {I_{17} } & {I_{17} } & {I_{17} } & {I_{17} } \end{array}} \\ {C_{85}^{0} +C_{85}^{24} +C_{85}^{40} +C_{85}^{71} +C_{85}^{84} } \\ {C_{85}^{1} +C_{85}^{49} +C_{85}^{58} +C_{85}^{81} +C_{85}^{84} } \\ {C_{85}^{3} +C_{85}^{14} +C_{85}^{32} +C_{85}^{78} +C_{85}^{84} } \\ {C_{85}^{16} +C_{85}^{33} +C_{85}^{50} +C_{85}^{67} +C_{85}^{84} } \end{array}\right). 
\end{multline} 

Here, $H_{i} $ represents circulant size 85 with weight 5, determined by the number of circulant terms with weight 1. The suitability of this parity-check matrix as a representation of Nitrogen is readily confirmed using the following logic. For ease of geometric analogies, let's focus on the first 3 rows of the matrix $H$:

\begin{multline} 
\label{1.2)} 
H=\left(\begin{array}{c} {\begin{array}{ccccc} {I_{17} } & {I_{17} } & {I_{17} } & {I_{17} } & {I_{17} } \end{array}} \\ {H_{1} } \\ {H_{2} } \end{array}\right)=\\ \left(\begin{array}{c} {\begin{array}{ccccc} {I_{17} } & {I_{17} } & {I_{17} } & {I_{17} } & {I_{17} } \end{array}} \\ {C_{85}^{0} +C_{85}^{24} +C_{85}^{40} +C_{85}^{71} +C_{85}^{84} } \\ {C_{85}^{1} +C_{85}^{49} +C_{85}^{58} +C_{85}^{81} +C_{85}^{84} } \end{array}\right). 
\end{multline} 

In the context of the considered model, this matrix can be obtained by collapsing the matrix along one of the coordinates:

\begin{multline} 
\label{1.3)} 
\left(\begin{array}{c} {\begin{array}{ccccc} {C_{85}^{k} } & {C_{85}^{k} } & {C_{85}^{k} } & {C_{85}^{k} } & {C_{85}^{k} } \end{array}} \\ {\begin{array}{ccccc} {C_{85}^{0} } & {C_{85}^{24} } & {C_{85}^{40} } & {C_{85}^{71} } & {C_{85}^{84} } \end{array}} \\ {\begin{array}{ccccc} {C_{85}^{1} } & {C_{85}^{49} } & {C_{85}^{58} } & {C_{85}^{81} } & {C_{85}^{84} } \end{array}} \end{array}\right)\to  \\ \left(\begin{array}{c} {\begin{array}{ccccc} {I_{17} } & {I_{17} } & {I_{17} } & {I_{17} } & {I_{17} } \end{array}} \\ {C_{85}^{0} +C_{85}^{24} +C_{85}^{40} +C_{85}^{71} +C_{85}^{84} } \\ {C_{85}^{1} +C_{85}^{49} +C_{85}^{58} +C_{85}^{81} +C_{85}^{84} } \end{array}\right) .
\end{multline}

Consider a scenario where five indistinguishable particles are evenly positioned from the center, described by spherical coordinates $\left(r,\varphi ,\theta \right)$ within a matrix. Now, let's explore the electron clouds of a carbon atom within the proposed model. Carbon carries a charge of 6, distributed as 2 electrons in the first energy level and 4 electrons in the second. A code matrix, presented in a general form, is constructed as follows:

\begin{equation} \label{1.4)} 
\left\{\begin{array}{c} {r} \\ {\varphi } \\ {\theta } \end{array}\right\}\left(\begin{array}{cccccc} {C_{k}^{r_{1} } } & {C_{k}^{r_{1} } } & {C_{k}^{r_{2} } } & {C_{k}^{r_{2} } } & {C_{k}^{r_{2} } } & {C_{k}^{r_{2} } } \\ {C_{k}^{\varphi _{1} } } & {C_{k}^{\varphi _{2} } } & {C_{k}^{\varphi _{3} } } & {C_{k}^{\varphi _{4} } } & {C_{k}^{\varphi _{5} } } & {C_{k}^{\varphi _{6} } } \\ {C_{k}^{\theta _{1} } } & {C_{k}^{\theta _{2} } } & {C_{k}^{\theta _{3} } } & {C_{k}^{\theta _{4} } } & {C_{k}^{\theta _{5} } } & {C_{k}^{\theta _{6} } } \end{array}\right) .
\end{equation} 

To determine the numerical characteristics of the circulants, specific conditions are defined. Specifically, for collapsing the matrix along the radii, it is necessary to observe the multiplicity of the order of the circulant to the radii and their sum:

\begin{equation} \label{1.5)} 
\begin{array}{l} {k\mod{r_{1}} =0}, \\ {k \mod{r_{2} =0}},
\\ {k \mod{\left(r_{1} +r_{2} \right)
}=0}. \end{array} 
\end{equation} 

In a more general case, expressed as:

\begin{equation} \label{1.6)} 
\begin{array}{l} {k \mod{r_{i} =0}}, \\ {k \mod{\sum _{i=1}^{N}r_{i}  =0} }. \end{array} 
\end{equation} 

Or, in the most general case (e.g., when the radii are prime numbers),  ${a,r_{i} \in {\mathbb N}}$:

\begin{equation} \label{1.7)} 
\begin{array}{l} {k \mod{\left(a\left(\sum _{i=1}^{N}r_{i}  \right)\prod _{i=1}^{n}r_{i}  \right)}}. \end{array}
\end{equation}

The presented mathematical framework introduces a cycle-based gauge, denoted as TS(a,0) or the Schrödinger-Heisenberg-Bohr-Fossorier Electron Cloud Gauge (SHBF Cycle Gauge), represented by the equation \cite{Fossorier04}:

\begin{equation} \label{1.9)} 
\left(\sum _{i=1}^{N-1}\Delta _{ji,ji+1} \left(l_{i} \right) \right) \mod {k}=0.
\end{equation} 

This equation ensures that the sum of shifts of the circulant along the row (or column) is a multiple of the dimension of the circulant. A specific matrix meeting this criterion is given by:

\begin{equation} \label{1.10)} 
\left\{\begin{array}{c} {r} \\ {\varphi } \\ {\theta } \end{array}\right\}\left(\begin{array}{cccccc} {C_{48}^{24} } & {C_{48}^{24} } & {C_{48}^{36} } & {C_{48}^{36} } & {C_{48}^{36} } & {C_{48}^{36} } \\ {C_{48}^{1} } & {C_{48}^{7} } & {C_{48}^{13} } & {C_{48}^{19} } & {C_{48}^{25} } & {C_{48}^{31} } \\ {C_{48}^{23} } & {C_{48}^{17} } & {C_{48}^{47} } & {C_{48}^{41} } & {C_{48}^{35} } & {C_{48}^{29} } \end{array}\right) .
\end{equation} 

It is crucial that the condition along the row is satisfied when the shift of all circulants of the row changes simultaneously by a multiple of:

\begin{equation} \label{1.11)} 
S=k/N .
\end{equation} 

Following this, the matrix is collapsed along the radii, resulting in a structure denoted as \(H\):

\begin{multline} \label{1.12)} 
H=\left(\begin{array}{c} {\begin{array}{cccccc} {C_{8}^{0} } & {C_{8}^{0} } & {C_{8}^{2} } & {C_{8}^{2} } & {C_{8}^{2} } & {C_{8}^{2} } \end{array}} \\ {C_{48}^{1} +C_{48}^{7} +C_{48}^{13} +C_{48}^{19} +C_{48}^{25} +C_{48}^{31} } \\ {C_{48}^{23} +C_{48}^{17} +C_{48}^{47} +C_{48}^{41} +C_{48}^{35} +C_{48}^{29} } \end{array}\right)  \to  \\ \left(\begin{array}{c} {\begin{array}{cccccc} {I_{8} } & {I_{8} } & {I_{8} } & {I_{8}} & {I_{8}} & {I_{8} } \end{array}} \\ {C_{48}^{1} +C_{48}^{7} +C_{48}^{13} +C_{48}^{19} +C_{48}^{25} +C_{48}^{31} } \\ {C_{48}^{23} +C_{48}^{17} +C_{48}^{47} +C_{48}^{41} +C_{48}^{35} +C_{48}^{29} } \end{array}\right).
\end{multline} 

The first row, where the collapse was performed, consists of 6 circulants of size 8 and weight 1. Two shift 0 circulants correspond to the first energy level, and four shift 2 circulants correspond to the second energy level, representing the electron configuration \(1s^{2} ;2s^{2} 2p^{2}\).

Lines 2 and 3 are circulants of weight 6 and size 48. The parity-check matrix of the quasi-cyclic LDPC code from Fossorier-Shu Lin's work corresponds to electron clouds in the ground state of a Nitrogen and Carbon atom under Spherical Topology and SHBF Cycle Gauge.

\begin{statement}  The carbon and nitrogen state diagram can be interpreted as a transition between states of matter under different temperatures (SNR). Such QC-LDPC codes representation allows greatly reducing the complexity of physical and chemical properties of matter. The existence of matter under different environmental properties is equivalent to the existence of a code on graphs with the required capacity properties. Moreover, Bethe-permanent allows simplifying it to linear complexity.
\end{statement}

\begin{figure}[t]
	\begin{center}
            \includegraphics[width=\columnwidth]{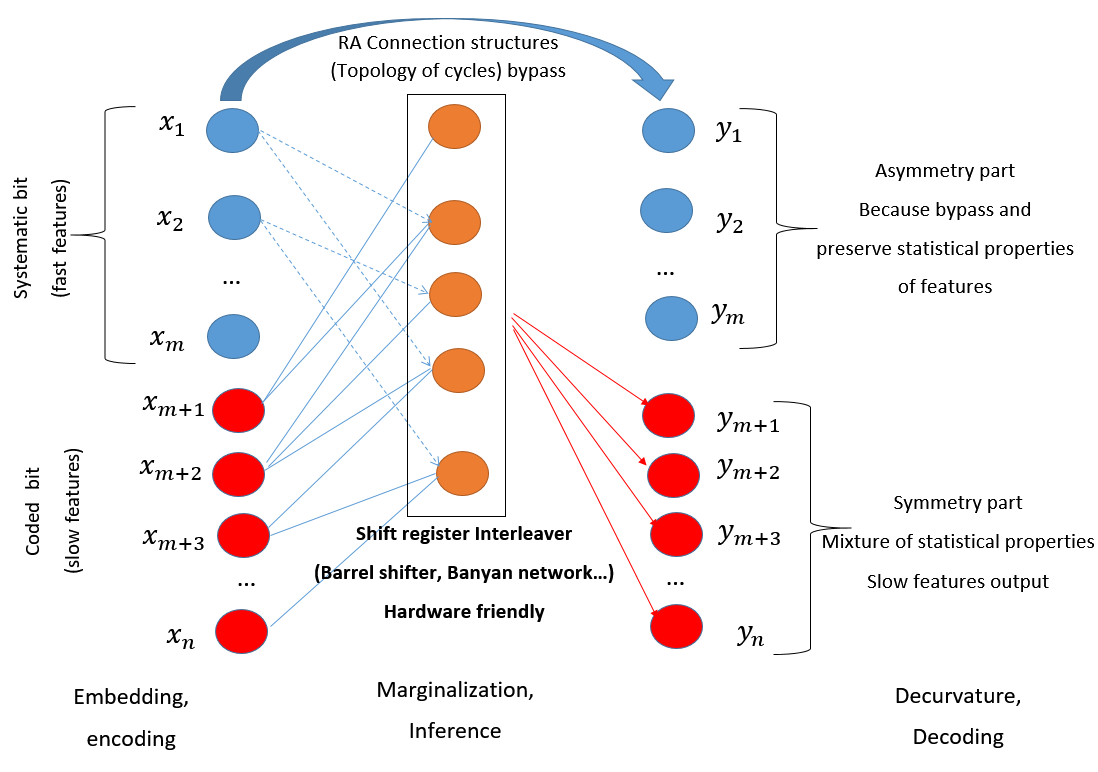}
	\end{center}
	\caption {Repeat Accumulate Code on the graph non-linear low dimensional toric and hyperbolic circular topology embedding as combination of symmetry and asymmetry(loss surface lattice/Hessian), \textbf{fast features avoid non-linear transformation} }
	\label{fig:images0099999}
\end{figure}

 \begin{figure}
\centering
\includegraphics[width=\columnwidth]{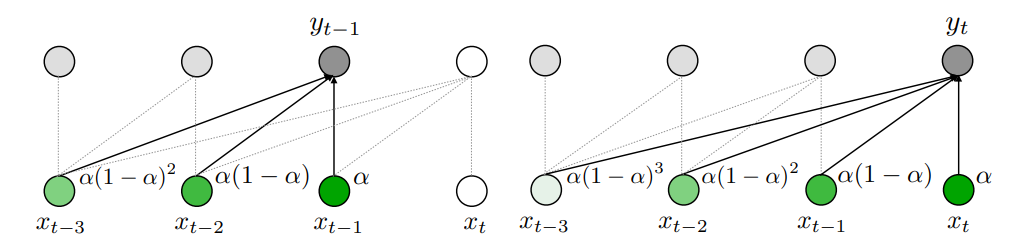} 
  \caption{Illustration of the exponential moving average (EMA) approach, which averages the input values X with weights decaying exponentially over timesteps. Equivalent to GeIRA QC-LDPC codes}
  \label{Mega}
\end{figure}

\section{ Demonstrating Equivalence: Generalized Repeat Accumulate Codes and Cage-Graph QC-LDPC Codes in NLP Transformer DNN}

Let's examine the application of an LDPC code parity-check matrix, denoted as $H$, in the context of encoding (embedding) a message (Tensor) $m$ for error correction. It's noteworthy that the nonlinear mapping and encoding can be viewed as embedding in machine learning (ML).

\begin{equation} 
x = H \times m,
\end{equation} 

Here, $m$ represents the transmitted message, $H$ is the parity-check matrix, $G$ is the generator matrix, and $x$ is the encoded message.

The encoding process satisfies the zero syndrome condition:

\begin{equation} 
G = H^{-1}, \quad H \times x = 0.
\end{equation} 

If the generator matrix of the code, by definition, satisfies $G = H^{-1}$, then:

\begin{equation} 
H \times G^T = 0.
\end{equation}
 The parity-check matrix of QC-LDPC codes for Repeat Accumulated (RA) codes can be expressed as:
\begin{equation} 
H = [ \, H_1 | H_2  ] .
\end{equation}
Here, $H_1$ denotes the parity-check part, which, when multiplied by systematic bits (analogous to features in DNN), is directly embedded through fast connections. Simultaneously, $H_2$ represents the special structures of the RA submatrix.For instance, if the RA parity-check matrix is defined as follows: 

\begin{equation} 
H_{2} =\left[\begin{array}{cccccc} {I} & {I} & {0} & {...} & {0} & {0} \\ {0} & {I} & {I} & {...} & {0} & {0} \\ {0} & {0} & {I} & {...} & {0} & {0} \\ {\vdots } & {0} & {0} & {...} & {0} & {0} \\ {\vdots } & {\vdots } & {\vdots } & {...} & {\vdots } & {\vdots } \\ {\vdots } & {\vdots } & {\vdots } & {...} & {I} & {0} \\ {0} & {0} & {0} & {...} & {I} & {I} \\ {0} & {0} & {0} & {...} & {0} & {I} \end{array}\right]  ,
\end{equation} 
where $I$ represents the identity circulant matrix of size $L \times L$. The inverse matrix $H_{2}^{-1}$ becomes triangular, and the encoding transforms into $G \times x$, as illustrated through the use of a convolutional accumulator (Fig. \ref{fig:images00888}, \cite{RiUr01,DivMcel98,Sarah10}), or within DNN embedding accomplished through dynamic convolution (involution and convolution features) (Fig. \ref{fig:images00999}, \cite{Sch91,WuFa19,LiHu21}).

The Mega and Mega-chunk Attention models are built upon a Generalized Irregular Repeat Accumulate (GeIRA) protograph, as outlined in ~\cite{DivMcel98, Li05}. In these models, the Exponential Averaging Alpha filter is employed to propel features forward without transformation, ensuring robustness to high temperatures and a quicker recovery. Additionally, we demonstrate that such Repeat Accumulate (RA) codes on graph structures can be generalized to irregular RA codes (GeIRA). This is exemplified through the Transformer MEGA architecture, which is represented by GeIRA codes on the graph, as depicted in \cite{DivMcel98, Li05}, Fig. \ref{Mega}.

 \begin{figure}
\centering
\includegraphics[width=\columnwidth]{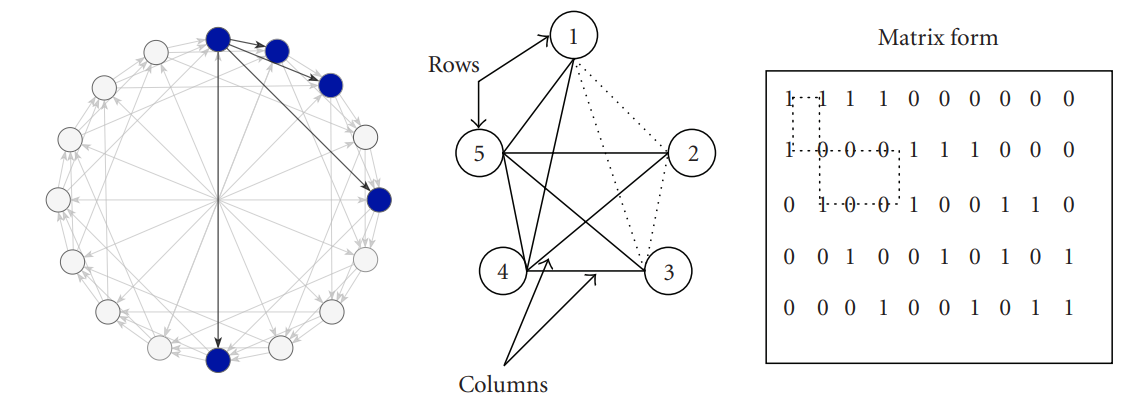} 
  \caption{Chord protocol weight 4 cage-graph. Distance graph and related LDPC code parity-check matrix. ChordMixer Transformer DNN Architecture }
  \label{Chord}
\end{figure}

A cutting-edge attention architecture (according LRA challenge \cite{LRA21}) in the realm of long-range interactions is introduced in the work by ~\cite{Kha22}. This architecture is built upon the P2P Chord protocol, as outlined in \cite{Stoica2001} (refer to \ref{Chord}, left). ChordMixer leverages Cage graphs as distance graphs to formulate its attention mechanism, as elucidated in the study by ~\cite{Ma07}. The utilization of Cage graphs enables ChordMixer to design the attention mechanism in a manner that is analogous to the  LDPC codes.

 \begin{figure}
\centering
\includegraphics[width=\columnwidth]{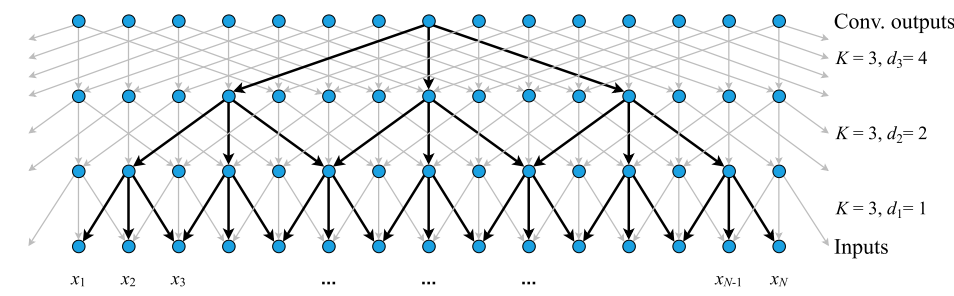} 
  \caption{The architecture of symmetric dilated convolutions. Computation graph of regular weight 3 codes. CDIL Transformer DNN Architecture ~\cite{Lei23}}
  \label{CDIL}
\end{figure}

Finally, we examine the dilated Convolutional model ("CDIL") introduced in the publication by ~\cite{Lei23}. This convolutional code is founded on the 2P Chord protocol and is categorized as a column weight 3 convolutional code, as illustrated in Fig. \ref{CDIL}. The theoretical characterization of the trapping set for such a code was comprehensively provided by ~\cite{Vasic09}, proposing an efficient decoding method based on \cite{Xi21}.

\begin{figure}[t]
	\begin{center}
            \includegraphics[width=\columnwidth]{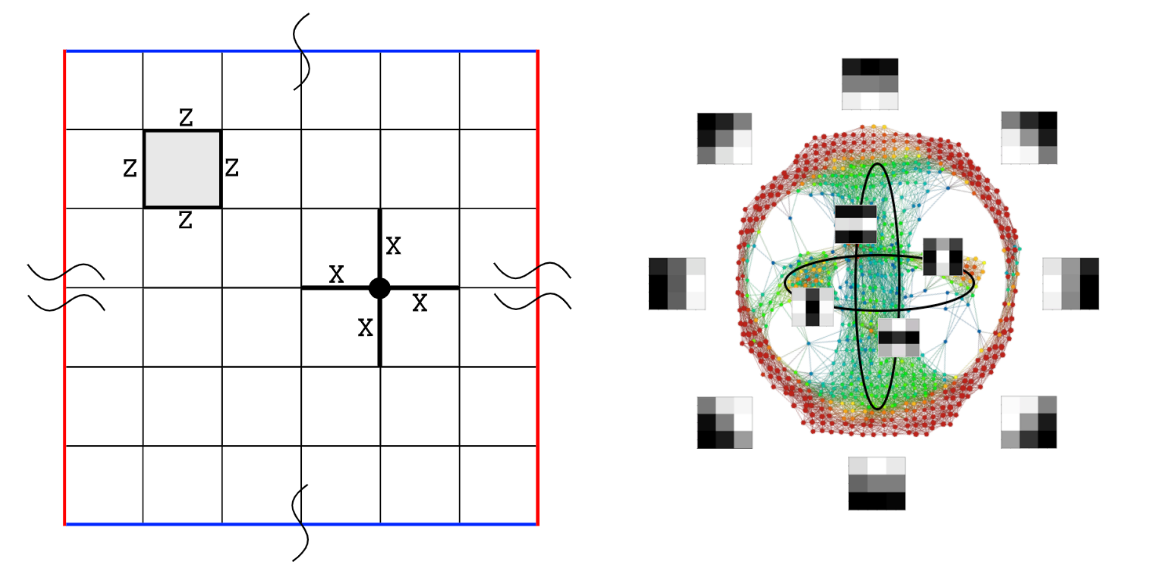}
	\end{center}
	\caption { Left side: Kitaev's Toric 2D  code \cite{Kitaev03}. Right side: Natural image Topology First layer, CIFAR-10,\cite{CaGa20}  }
	\label{Fig555}
\end{figure}

 \begin{figure}
\centering
\includegraphics[width=\columnwidth]{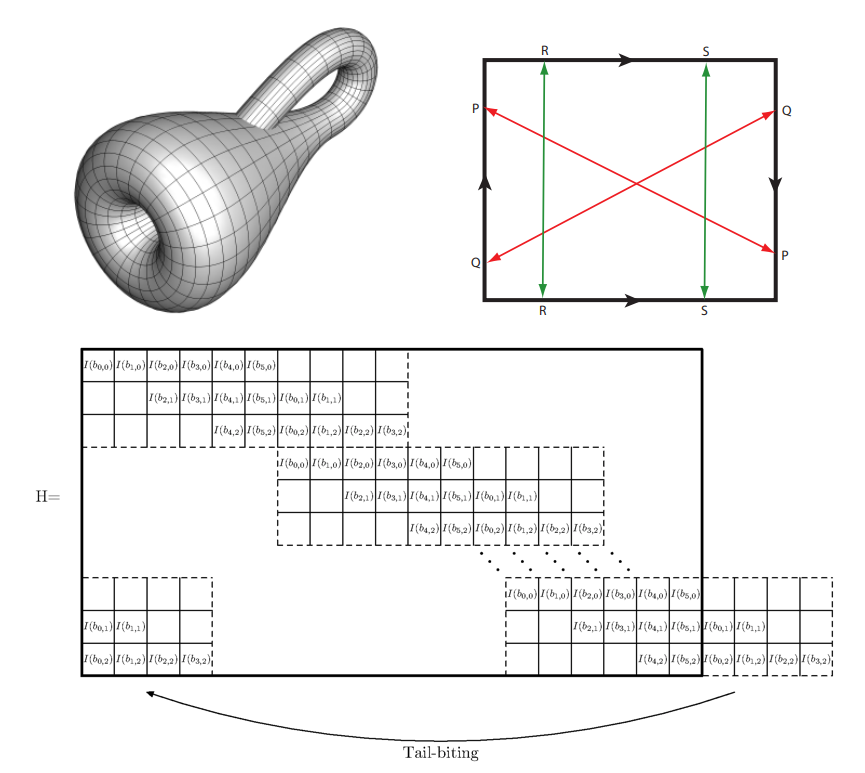} 
  \caption{Klein bottle and Zigangirov's Spatially-Coupled Convolutional QC-LDPC code \cite{FelZi99}  as half-twisted Hyperbolic Torical Topology}
  \label{Fig333}
\end{figure}

An intriguing observation lies in the link to Torical quantum codes under Ising models, established through the Calderbank-Shor-Steane (CSS) stabilizer code construction \cite{CalShor96,Kitaev03,Gleb21}, as illustrated in Figure \ref{Fig555} on the left. This connection with the ground states of the Ising model has been substantiated in previous subsections. The study  \cite{CaGa20} utilized topological data analysis to gain insights into the computations performed by convolutional deep neural networks. Analyzing the internal states of these networks, they modified computations to enhance speed and generalization across various datasets.

Their analysis also resulted in a new geometry for sets of features on image datasets (MNIST, CIFAR-10, SVHN, ImageNet, etc.). This allowed them to construct analogues of CNNs for various geometries, including graph structures derived from topological data analysis.

Building on previous research, such as Carlsson et al.'s study \cite{CaGa20}, which demonstrated that local patches in images concentrate around a primary circle and can be modeled by algebraically defined functions (Fig. \ref{Fig555} on the right), their work led to the construction of features for images admitting a Klein bottle Topology.

The natural generalization of Torical topology under infinite-size circulant (automorphism) is the half-twisted Klein bottle topology equivalent to convolutional codes: Spatially-coupled Tailbitted QC-LDPC codes, Polar Spatially-Coupled codes, and Serial Turbo codes (Fig. \ref{Fig333}). This connection elucidates the effectiveness of convolutional codes in computer vision tasks. Moreover, it demonstrates the link between topological analysis and the calculation of topological invariants, such as Betti numbers, indicating the number of multidimensional voids formed by code and pseudo-code words of trapping sets (TS enumerators).

\textbf{Proposition 3:} The research paper \cite{TanSSFCost04} establishes a connection between QC block codes and spatially coupled QC convolutional codes with limited memory termination. This connection, observed under the Ising model, is rooted in the klein bottle twisted topology present in block QC codes featuring a toric (hyperbolic) topology. Exploiting this relationship enables the creation of a dynamic convolution aligned with the mentioned topology, incorporating low-complexity non-linear multiplicative correlation. This dynamic convolution involves fast and slow features, systematic bits (fast feature, pass forward without changes), and coded bits in channel codes, utilizing highly efficient convolutional accumulator processing commonly employed in wireless encoders, such as QC-LDPC codes \cite{Li05,DivMcel98,RiUr01,Sarah10}. This aligns with the TDA analysis findings by Carlsson \cite{CaGa20}, as originally proposed by Schmidhuber (1991) and described in papers \cite{Sch91,WuFa19,LiHu21}.

The integration of QC-LDPC block and convolutional codes, along with their relationship to Polar and Turbo codes, through serial processing on Klein bottle and Torus Topology, significantly contributes to enhancing the quality of structured codes on graph models with high automorphism. This integration plays a crucial role in various applications, spanning material science, complexity theory, and efficient implementations of deep neural networks. A key metric for assessing the gauge quality of any dynamic system is the Trapping sets enumerator, encompassing codewords and pseudocodewords. It is asserted that the Trapping sets enumerator serves as a fundamental metric for evaluating the gauge, defining the loss surface (energy landscape) of spherical and toroidal dynamic systems, including Ising models, DNNs, control systems, differential equations, social networks, and other related topology graph models and tensor networks.

\textbf{Proposition 4:} Spatially-Coupled Convolutional codes resembling a Klein bottle topology can be constructed by adopting a suitable graph structure, achieved by twisting opposite edges on a torus. The link between block and convolutional QC-LDPC codes is evident in \cite{TanSSFCost04}, while the connection between Polar codes and LDPC codes, post-sequential processing with a successive cancellation scheduler, is discussed in \cite{Fer13,Fos15}. A similar connection is showcased between Turbo codes and QC-LDPC codes under serial processing \cite{JiaPsP06}, providing a potential approach for hardware-efficient optimal embedding using Codes on the graphs TS enumerators and related Topology Invariant from Topology Data Analysis.

Non-linear sparse factorization serves as a low-dimensional embedding capturing the structure of high-dimensional data. Leveraging topology data analysis and information geometry, we can analyze data curvature, facilitating efficient and accurate tasks like classification, interpolation, and regression. Stochastic dynamic systems accommodate inherent data uncertainty, ensuring robust predictions in real-world scenarios. The spectral gap, tied to topological properties like Bethe tree or lattice, Bethe permanent (pseudocodewords TS(a,b)) \cite{Vo13,RoxanaPseudoBound}, and permanent (codewords TS(a,0)) \cite{Mac01,RoxanaCodeBound}, along with the graph spectral bound \cite{Ta01,Mi05,Lu06}, are crucial. Neural networks with minimal clustering coefficients, particularly entangled nets (optimized Trapping set enumerator), exhibit promising traits for constructing high-capacity and high-performance artificial neural networks \cite{Lu06}. Improving spectral properties, especially bethe permanent and permanent, enhances Trapping set enumerators in graph models, leading to increased capacity and improved performance in DNNs.

\begin{figure}[t]
	\begin{center}
            \includegraphics[width=\columnwidth]{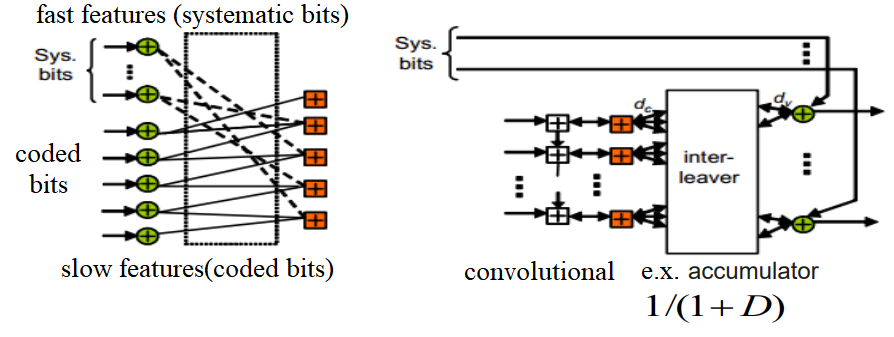}
	\end{center}
	\caption {Repeat Accumulate Code on the graph structures, fast part systematic bits(features for DNN) and repeat accumulated coded bits, \cite{Li05,RiUr01,DivMcel98,Sarah10}  }
	\label{fig:images00888}
\end{figure}

\begin{figure}[t]
	\begin{center}
            \includegraphics[width=\columnwidth]{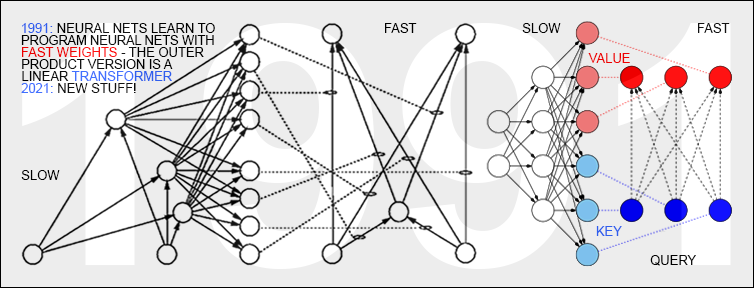}
	\end{center}
	\caption {The dynamic approach combines: convolutional and block codes by incorporating attention and convolution within a deep neural network (DNN) proposed by Schmidhube at 1991, \cite{Sch91}. RA codes example of such approach }
	\label{fig:images00999}
\end{figure}

\section{CONCLUSION}

This paper introduces a pioneering approach to achieve equilibrium in the ISING Hamiltonian when dealing with unevenly distributed charges on an irregular grid. By leveraging QC-LDPC codes and the Boltzmann machine, our method involves expanding the system dimensionally, replacing charges with circulants, and expressing distances through circulant shifts. This systematic mapping of the charge system onto a space transforms the irregular grid into a uniform configuration, applicable to Torical and Circular Hyperboloid Topologies. The fundamental definitions and notations related to QC-LDPC Codes, Multi-Edge QC-LDPC codes, and the Boltzmann machine are covered. The paper delves into the marginalization problem in code on the graph probabilistic models for evaluating the partition function, exploring exact and approximate estimation techniques. Rigorous proof is provided for the attainability of equilibrium states for the Boltzmann machine under Torical and Circular Hyperboloid, laying the foundation for the application of our methodology. Practical applications of our approach are investigated in Finite Geometry QC-LDPC Codes, particularly in Material Science. The paper further explores its effectiveness in the domain of Natural Language Processing Transformer Deep Neural Networks, examining Generalized Repeat Accumulate Codes and Cage-Graph QC-LDPC Codes, Convolutional codes. The versatile and impactful nature of our topology-aware equilibrium method is demonstrated across diverse scientific domains, transcending specific section delineations.

\end{document}